\documentclass[11pt]{article}
\usepackage{natbib}
\usepackage{preamble-vix}
\usepackage{authblk}
\usepackage{fullpage}

\title{How to Sell High-Dimensional Data Optimally}
\author{Andrew A. Li}
\author{R. Ravi}
\author{Karan Singh}
\author{Zihong Yi}
\author{Weizhong Zhang}
\affil{Tepper School of Business, Carnegie Mellon University, Pittsburgh}
\affil[ ]{\texttt{\{aali1,ravi,karansingh\}@cmu.edu,\{zihongy,weizhong\}@andrew.cmu.edu}}
\date{}

\begin{document}
\maketitle
\begin{abstract}
Motivated by the problem of selling large, proprietary data, we consider an information pricing problem proposed by~\cite{bergemann2018design} that involves a decision-making buyer and a monopolistic seller. The seller has access to the underlying state of the world that determines the utility of the various actions the buyer may take. Since the buyer can make better decisions given more accurate assessments of the state, the seller can promise the buyer supplemental information at a price. As the seller may not be perfectly informed about the buyer's private preferences, we frame the problem of designing data products as one where the seller designs a revenue-maximizing menu of statistical experiments. 

Previous work by~\cite{cai2021sell} showed that an optimal menu can be found in time polynomial in the size of the state space, whereas we observe that the latter quantity is naturally exponential in the dimension of the data. We propose an algorithm which, given only sampling access to the state space, provably generates a near-optimal menu with a number of samples independent of the size of the state space. We then analyze a special case of high-dimensional Gaussian data, showing that (a) it suffices to consider scalar Gaussian experiments, (b) the optimal menu of such experiments can be found efficiently via a semidefinite program, (c) full surplus extraction occurs if and only if a natural separation condition holds on the set of potential preferences of the buyer, and (d) deterministic experiments suffice for high-dimensional data. We empirically validate our algorithm on a routing example using real traffic data, finding that 95\% of the optimal revenue is attained with just 40 samples from the state space.
\end{abstract}

\section{Introduction} \label{sec:intro}
Consider the problem of selling proprietary data so as to maximize revenue. The prevalence and importance of this problem in practice are by this point firmly-established. In fact, the Electronic Frontier Foundation identified 1,400 separate data {\em brokers} operating in the United States as of 2025,\footnote{\url{https://www.eff.org/document/appendix-b-databrokerfullregistry2025}} to say nothing of the number of individual {\em sellers} on these platforms. Secondary data in the financial industry (so-called ``alternative data'') alone constituted an \$18 billion market in 2025, and is expected to grow 60\% yearly through 2030.\footnote{\url{https://www.grandviewresearch.com/industry-analysis/alternative-data-market}} The market for data for B2C marketing is higher by at least an order of magnitude.\footnote{\url{https://hbr.org/2020/05/buying-consumer-data-tread-carefully}} And most recently, the training of large-scale machine learning models has created another market for data, and the race for ever-increasing size and sophistication in these models all but guarantees that the value of high-quality data will rise in step.

Any principled approach to this problem must deal with a set of challenges that include, at the very least, (a) designing a set of potential data ``products,'' (b) identifying the value that potential buyers place on these products, and (c) selecting the optimal subset of products to offer, together with their corresponding optimal prices. Much of the prior work on selling data addresses strict subsets of these challenges, often in the context of stylized or parametric models. 

A line of work that began with \cite{bergemann2018design} allows all of these challenges to be modeled cleanly and simultaneously. That model treats ``information'' (of which data is a special case) as valuable to a buyer insofar as it improves their eventual selection of a utility-maximizing action. There is assumed to be an abstract true {\bf state} $\omega \in \Omega$, which is unknown to buyers and governs the resulting utility of each abstract {\bf action} $a \in \mathcal{A}$. Buyers of different {\bf types} have their own utility functions and their own prior belief distributions on the true state. The space of potential ``products'' to sell to them consists of statistical {\bf experiments}, which is entirely generic in representing the manner in which buyers can update their beliefs on the true state. 

It is worth introducing a few applications at this point to make the discussion more concrete:
\begin{itemize}
\item {\em Data-driven decision-making:} Buyers who face a concrete decision-making task my be interested in reducing uncertainty with data. As one example, which we will make concrete in our experiments: live traffic data can be sold to buyers who intend to solve (implicitly or explicitly) different routing problems (e.g.~shortest path, traveling salesman, vehicle routing). In this case, the actions are the possible routes, the true state is the set of (unknown) travel times on every edge of the network (or every road, say), and greater utility corresponds to a shorter total travel time. Two other standard application areas here are financial investors optimizing a portfolio when future returns are uncertain, and marketers deciding which potential customers to target when receptiveness to such efforts is uncertain.

\item {\em Training ML models:} Buyers may be using data to train machine learning models. In this case, the true state is the data itself (encoded numerically), the set of actions is a set of machine learning models (encoded parametrically) from which the buyer will eventually select, and  utility corresponds to some measure of accuracy.

\item {\em Features for ML predictions:} Another use of data lies in machine learning inference. Buyers may use the data as feature vectors while using, for example, a linear regression model to estimate an uncertain quantity of interest. Here, the action is scalar, and the true state is the data itself, and greater utility results from ever more accurate estimates of the uncertain target. We will analyze a special version of this application in detail.
\end{itemize}
In addition to illustrating the wide-ranging applications of the model we will study, these examples underscore a critical limitation of work to date: {\em scalability with respect to the state space $\Omega$}, which is often high-dimensional as in the applications above. 
In particular, this model is only useful if the resulting {\em algorithmic} problem of selecting and pricing data products is solvable but, unfortunately, previous work in this regard is insufficient. The original work of \cite{bergemann2018design} solved smaller special cases, including $|\Omega|=2$, and general $|\Omega|$ but limited to two types of buyers. \cite{cai2021sell} then proposed a linear programming algorithm, although its size scales polynomially with $|\Omega|$. We contend that $|\Omega|$ is the limiting factor in practice, as it grows exponentially in the dimension of the data.\footnote{Incidentally, \cite{cai2021sell} discuss the same shortest-path application as here, but in their case they describe $\Omega$ as a two-element set that encodes ``low'' and ``high'' traffic.}

\vspace{1em}
\subsection{Main Contributions}

Thus motivated, we study the model of \cite{bergemann2018design} with the aim of developing an algorithm for large state spaces. Our contributions are as follows: 

\vspace{1em}
\noindent
{\bf 1. An Efficient Algorithm for Data Pricing:} The primary contribution of this paper is an algorithm for data pricing that obtains near-optimal revenue with a runtime that is {\em independent} of the potentially infinite state space:

\vspace{1em}
\begin{quote}
\noindent {\bf Theorem~\ref{thm:main} (Informal).} {\em Given samples of the state space drawn according to the buyers' belief distributions, there is an algorithm that computes a near-optimal menu of experiments and prices with high probability, and both its runtime and sample complexity are independent of the size of the state space and polynomial in the number of actions.}
\end{quote}
\vspace{1em}
Table \ref{tab:comparison} highlights the key contrasts between this result and previous algorithmic work. For example, in addition to scalability with respect to the state space, our algorithm's polynomial dependence on the action space compares well to the choice in existing work between (a) polynomial dependence and (b) independence at the cost of super-exponential dependence on the state space.

\begin{table} 
\centering
\begin{tabular}{@{} cccc @{}}
\toprule
{\bf Paper} & {\bf Buyer Types} & {\bf State Space Size} & {\bf Runtime} \\ \midrule 
\cite{bergemann2018design} & $n$ & 2 & $O(n)$ for $|\mathcal{A}| = 2$\\
\cite{bergemann2018design} & 2 & Finite & $O(|\Omega||\mathcal{A}|)$ \\
\cite{cai2021sell} & $n$ & Finite & $\mathrm{poly}(n,|\Omega|,|\mathcal{A}|)$\\
\cite{cai2021sell} & $n$ & Finite & $\mathrm{poly}(n,|\Omega|^{|\Omega|^2},1/\epsilon^{|\Omega|^2})$*\\
This Paper & $n$ & Infinite & $\mathrm{poly}(n, |\mathcal{A}|, 1/\epsilon)$**\\
\bottomrule
\end{tabular}
\caption{Comparison to prior work. Reported runtimes that include $\epsilon$ are for algorithms which approximate the optimal value up to an additive factor. 
*Requires an oracle that produces the expected-utility-maximizing action for any given distribution over states. **Requires a sampling oracle over the state space.}
\label{tab:comparison}
\end{table}




\vspace{1em}
\noindent
{\bf 2. High-Dimensional Gaussian Data:} Our second contribution is a sharpening of our results for an important special case where the state space is $\R^d$ and the state $\omega$ follows a Gaussian distribution. Here, each buyer type has a private preference vector \( \theta_i\in \mathbb{R}^d \), which differentially weighs the various features (coordinates) of \( \omega \) along which they wish to estimate the unknown state. The utility of type $i$ for action $a$ is $u^i(\omega, a) = - (\theta_i^\top \omega - a)^2$. For this setting, we have the following results.
\begin{enumerate}[itemsep=0em]
    \item We show (in \Cref{thm:gauexp}) that it is sufficient, without any loss of revenue, to consider scalar Gaussian experiments of the form $\mathcal{N}(v^\top \omega, \sigma^2)$, which project the state $\omega$ along a direction before adding noise, although this direction may not align with any of the preference vectors.
    \item  We provide (via \Cref{thm:surplus}) an intuitive condition characterizing when the seller can extract the full surplus from the buyer. This happens if and only if, for each type $i$, its preference vector $\theta_i$ is longer than the shadow that other $\theta_j$'s cast on it. We also show (in \Cref{thm:deter}) for high-dimensional states that there exists a revenue-optimal menu composed entirely of deterministic experiments ($\sigma_E=0$).
    \item We formulate a semidefinite program (in \Cref{thm:sdp}), solvable in time scaling polynomially in the number of buyers and the {\em dimension} of the underlying state space, that computes such a revenue-maximizing menu composed entirely of scalar Gaussian experiments. 
\end{enumerate}


\noindent This special case reveals important qualitative properties of the broader model we study, including the relative value of offering multiple data products to extract revenue from heterogeneous buyers, and the sufficiency of deterministic data products (i.e.~those which do not introduce ``artificial'' randomness).

\vspace{1em}
\noindent
{\bf 3. Experimental Validation:} Our final contribution is an empirical validation of our algorithm (in the general setting of Contribution 1) on a routing example using real traffic data from the publically-available {\em METR-LA} traffic dataset. Our results show that our algorithm, which relies on sampling from the state space, achieves near-optimal revenue given only a small number of samples. Indeed just 40 samples were needed to attain 95\% of the optimal revenue in our experiment (see \Cref{fig:1}).

\vspace{1em}
We conclude this section with a brief review of related work. We then introduce the model and problem in Section~\ref{sec:prelim}. In Section~\ref{sec:sampleLP}, we present our algorithm and prove our main result. In Section~\ref{sec:gaussian}, we analyze the specialized Gaussian setting. We present our experimental results in \Cref{sec:experiments}, and conclude in Section~\ref{sec:conclusion}.

\vspace{1em}
\subsection{Related Literature}
The scope of our work intersects with the areas of monopoly pricing,  information design, information pricing, and data markets.

{\bf Monopoly Pricing:}
A critical but distinct issue is the monopoly pricing problem associated with the sale of goods, where the buyer's valuation is derived from a known distribution. Complete characterizations are known in the single-dimensional setting~\citep{myerson1981optimal, riley1983optimal}, and for specific cases in multidimensional settings~\citep{fiat2016fedex, devanur2020optimal, haghpanah2015reverse, giannakopoulos2014duality}.  More broadly, classical multidimensional screening and multiproduct pricing models illustrate that once buyer heterogeneity is genuinely multidimensional, optimal tariffs can exhibit nontrivial geometric features such as exclusion regions and bunching, and may require rich menus even in stylized environments~\citep{armstrong1996multiproduct, rochet1998ironing}. Furthermore,~\cite{esHo2007optimal} explored the simultaneous design of mechanisms for selling an \emph{indivisible} good, where the auctioneer can, without prior observation, release additional signals about the item. 
Relative to goods, pricing for data/information differs significantly in two primary ways: First, while goods are typically indivisible, the pricing of information has more complexity due to its intangible nature. Second, buyers with varying preferences do not only assign different values to various products; their preferences also influence the ranking of these products. Consequently, the value of information inherently encompasses both a vertical element (the quality of the information) and a horizontal element (the position of the information).

{\bf Information Design:}
The primary goal in information design is for the designer to disclose signals about the state to influence the actions of the agents involved~\citep{bergemann2019information, kamenica2011bayesian, kamenica2019bayesian, alonso2016bayesian, dughmi2016algorithmic, neumann2019effective, guo2025selling}. A useful perspective emphasized in this literature is to view a disclosure policy as a commitment to an information structure, with the set of implementable outcomes disciplined by Bayesian plausibility; see, e.g., the Bayes correlated equilibrium formulation in~\cite{bergemann2016information}. This viewpoint clarifies how different signal structures shape agents’ best responses and provides a language for comparing disclosures.  A fundamental distinction arises in our model of information selling: unlike traditional information design, where the designer's revenue is typically derived from the actions taken by the buyers, in our model the seller's revenue is accrued directly from the payments made by the buyers.

{\bf Information Pricing:}
Two primary research streams have emerged: {\em ex ante} pricing and {\em ex post} pricing, in parallel to real-life information pricing models, specifically subscription-based (ex ante) and one-time payment systems (ex post).
In ex ante pricing, the seller commits to revealing information about the state $\omega$ without prior observation, while in ex post pricing, the pricing decision is contingent upon the state $\omega$ observed at the time of the transaction. 

{\em Ex-ante Pricing:}
\cite{admati1986monopolistic, admati1990direct} analyzed the sale of information to a continuum of ex ante homogeneous buyers, demonstrating that it is optimal to provide noisy and idiosyncratic information, thereby ensuring that the seller maintains a local monopoly. The context in their work is the rational expectations equilibrium, which requires interaction between buyers. This contrasts with our model where there is no interaction among buyers, and buyers are inherently heterogeneous due to their utilities. 
As mentioned earlier, we use the model introduced in \cite{bergemann2018design},
where the state space is discrete and the seller's information is conveyed through experiments~\citep{blackwell1951comparison, blackwell1953equivalent}. They provide an explicit construction of the optimal menu in the case of binary states and actions. Based on their framework, \cite{cai2021sell} proposed a linear program to calculate the revenue-maximizing menu. 
They also highlighted that when a best-action oracle is available, that is, an oracle that receives a distribution of states as input and returns the optimal action, there exists an algorithm capable of efficiently (in the number of actions) computing the revenue maximizing menu. 

{\em Ex-post Pricing:}
A distinct body of research focuses on ex-post pricing where the value function post-realization depends on two state variables - one held by the data seller and the other by the buyer~\citep{sundararajan2004nonlinear, arora2005pricing, babaioff2012optimal, liu2021optimal, chen2022selling}. In these settings, the seller can condition both the information disclosed and the payment rule on the realized signal, and the mechanism may involve interaction in order to elicit the buyer’s private information or to tailor recommendations. Ex-post pricing highlights substantial differences from ex-ante pricing, as it permits the seller to tailor experiments based on the realized state, thereby creating additional degrees of freedom that are unavailable under ex-ante commitment.

{\bf Data Markets / Information Intermediation:}
A complementary literature studies the organization of markets in which data (and information) is traded as a product, often through intermediaries, to improve downstream targeting and decision-making. In advertising environments, for instance, ~\cite{bergemann2015selling} model the sale of consumer ``cookies,'' highlighting how a data provider’s information product shapes buyers’ targeting decisions and, in turn, demand for data. More broadly, when data are nonrival, the resale and reuse of the same information can generate externalities that are central to intermediated data markets~\citep{bergemann2022economics}. A related perspective formalizes data products as market segmentations~\citep{yang2022selling}, underscoring that many practical data offerings can be viewed as structured information disclosures. These works are complementary to our approach: while they emphasize market organization and downstream strategic effects, our model abstracts the information product itself as a menu of statistical experiments, allowing us to study optimal design and—crucially for high-dimensional data—computationally efficient methods for constructing revenue-maximizing menus.

\section{Preliminaries}
\label{sec:prelim}

In this section, we describe the timeline of interactions between the buyer and the seller. Following certain simplifications that can be made without loss of generality on the space of messages that the buyer sends, we see that the revenue maximization problem can be recast as a convex program, albeit one whose size scales linearly with the number of states. Finally, to motivate the sale of differentiated data products, we characterize the range of advantages that the design of different data products has over pricing alone and its ability to overcome seller-side information asymmetry.

\vspace{1em}
\subsection{Model}
We begin by formalizing the model we study. The buyer has a state-dependent utility that they want to maximize by taking an action from an action space.
The buyer has a private {\em type} that captures their (ex-post) utility as a function of both the the action and the unknown state. Although the seller is not aware of the buyer's private type, they know that the type follows a known distribution. The seller determines a menu of experiments with corresponding prices to offer for sale to the buyer. Each experiment maps the state to a distribution over a space of signals. Attaching notation to these terms, we have:
\begin{itemize}
    \item The {\em state} $\omega$ is drawn from a distribution $\mu$, a publicly known common prior supported on a possibly uncountable but measurable state space $\Omega$. 
   
    \item The finite {\em action} space $\mathcal{A}$ is of size $m=|\mathcal{A}|$. 
        
     \item The buyer's {\em type} $i$, belonging to $[n]$, determines their utility function $u^i:\Omega \times \mathcal{A} \mapsto [0,1]$. The distribution over types is captured by $f$, which is a member of the $(n-1)$-dimensional probability simplex.
    \item A statistical {\em experiment} is described as $E = (S, \pi)$, where $S$ is the finite set of {\em signals}, and $\pi: \Omega \mapsto \Delta S$ is the signaling function that maps each state $\omega$ to a distribution $\pi(\cdot | \omega)$ over the signals. 
    The fact that $S$ is finite holds without loss of generality, as we will discuss momentarily.
    \item A {\em menu} of experiments can be written as $\mathcal{M} = \{\mathcal{E}, t\}$, where $\mathcal{E}$ is a set of experiments, and $t:\mathcal{E} \mapsto \mathbb{R}_{\geq 0}$ is the price function that denotes the price $t(E)$ of each experiment $E \in \mathcal{E}$.
\end{itemize}

The precise game between the buyer and the seller modeling the design, sale and eventual usage of the menu of experiments proceeds as follows: 
\begin{enumerate}
\item The seller posts a menu $\mathcal{M} = \{\mathcal{E}, t\}$.
\item The type $i$ of the buyer is determined.
\item The buyer chooses an experiment $E=(S,\pi) \in \mathcal{E}$ and agrees to pay a (randomized) price $t(E)$.
\item The state $\omega$ is realized, and the seller sends a signal $s$ that is drawn from $\pi(\cdot \mid \omega)$.
\item The price $t(E)$ is fully determined and paid.
\item The buyer updates their belief about the state $\omega$ based on the observed signal $s$ and chooses an optimal action $a$ maximizing their expected utility.
\end{enumerate}
In later sections, we will also describe how our approach can be easily adapted to settings where the price is realized and paid before the state is realized.

In the above setup, both the buyer and the seller are expected utility-maximizing agents. Thus, in the absence of any privileged information, a buyer of type $i$ has a baseline utility of $$u^i := \max_{a \in \mathcal{A}} \mathbb{E}_{\omega \sim \mu}[u^i(\omega, a)],$$ resulting from taking the action $$a^i \in \argmax_{a \in \mathcal{A}} \mathbb{E}_{\omega \sim \mu}[u^i(\omega, a)].$$

Now, if the same buyer receives the signal $s$ from experiment $E=(S,\pi)$, they revise their belief before selecting an action. Hence, assuming that the signal $s$ is sampled with a non-zero probability overall,
their conditional expected utility will be
\begin{align*}
u^i(s,E) &:= 
\max_{a\in \mathcal{A}} 
\mathbb{E}_{\omega \sim \mu|(s, E)}[u^i(\omega, a)] \\
&=
\max_{a\in \mathcal{A}}
\mathbb{E}_{\omega \sim \mu}\!\left[
 \frac{\pi(s \mid \omega)}
      {\mathbb{E}_{\omega'\sim\mu}[\pi(s \mid \omega')]}
 u^i(\omega, a)
\right],
\end{align*}
where $\mu|(s,E)$ reflects the buyer's posterior belief of the state $\omega$ having observed the signal $s$ generated from the experiment $E$. Thus, the buyer will select an action from the set $a^i(s,E)$ denoting the actions which achieve this maximum. 

Together with $\mu$, each experiment $E=(S,\pi)$ produces a marginal distribution $\pi \circ \mu$ on the signal space. We can use this to calculate the {\em value} of an experiment $E=(S,\pi)$ for a buyer of type $i$:
\begin{align*}
V(E,i) & := 
\mathbb{E}_{s\sim \pi\circ \mu}\bigl[u^i(s, E)\bigr] \\
&=  
\sum_{s \in S} 
\mathbb{E}_{\omega \sim \mu}\!\bigl[\pi(s \mid \omega)\bigr]\,
u^i(s,E) \\
&= 
\sum_{s \in S} 
\max_{a\in \mathcal{A}} 
\mathbb{E}_{\omega \sim \mu}
\!\bigl[\pi(s \mid \omega)\, u^i(\omega,a)\bigr].
\end{align*}
Note that since $V(E,i)$ is a point-wise maximum of linear functions, it is in fact convex in $\pi$.

\vspace{1em}
\subsection{The Revenue Maximization Problem}
The problem we study, which is the seller's problem, is to select the menu of experiments (including the corresponding prices) so as to maximize expected revenue. It may be unclear at first glance how exactly this problem can be addressed given the generality of the model: experiments can be composed from any mapping of the state space $\Omega$ to distributions over any signal space $S$, and a menu can be composed of any number of such experiments. Fortunately, two classic reductions, which can be made without loss of generality, simplify this problem.

First, by the revelation principle \citep{myerson1982optimal}, the seller can restrict their attention to direct mechanisms, where there are as many items on the menu as the number of types of the buyer. In addition, the seller can restrict to incentive compatible menus, where the menu has $n$ entries $\{E^1,\ldots,E^n\}$ such that each type $i$ favors buying $E^i$ over other experiments. Generically, a buyer of type $i$ is incentivized to purchase an experiment $E^i$ in isolation if and only if $$V(E^i, i) - t(E^i) \geq u^i.$$
Viewed as a constraint on the menu, this is typically referred to as {\em individual rationality (IR)}.
Moreover, every buyer selects the experiment that maximizes their own expected utility, so for any $i, j \in [n]$, we have
$$V(E^i, i) - t(E^i) \geq    V(E^j, i) - t(E^j).$$ These are the {\em incentive compatibility (IC)} constraints. We will abbreviate $t(E^i)$ by $t^i$ going forward.
 
Second, following a subtler application of the revelation principle in \cite{bergemann2018design}, the seller can further take $S=\mathcal{A}$, which means that the signal and action spaces are the same, and restrict their search to {\em responsive} menus without loss of generality. A responsive menu is one in which every signal in each experiment corresponds to the optimal action it results in for the intended buyer. Concretely, a responsive menu requires that $a\in a^i(a,E^i)$ for all types $i$ in $[n]$ and signals $a$ in $\mathcal{A}$. The intuition here is that, given a non-responsive experiment for type $i$, one can merge and relabel all signals that result in the same action for type $i$, without hampering incentive compatibility constraints. Responsiveness can be enforced through {\em obedience} constraints \citep{bergemann2019information} as
$$ \mathbb{E}_{\omega \sim \mu} [\pi(a \mid \omega) u^i(\omega,a)] = \max_{a'\in \mathcal{A}} \mathbb{E}_{\omega \sim \mu} [\pi(a \mid \omega) u^i(\omega,a')], \qquad \forall a\in \mathcal{A}, i\in [n].$$
Conveniently, since every signal from $E^i=(\mathcal{A}, \pi^i)$ in a responsive menu coincides with the utility-maximizing action for type $i$, $V(E^i, i)$ can further be written as a linear function of $\pi_i$ as
\begin{align*}
    \overline{V}(E^i,i), 
    = \sum_{a \in \mathcal{A}} \mathbb{E}_{\omega \sim \mu} [\pi^i(a \mid \omega) u^i(\omega,a)],
\end{align*}
where we denote $\overline{V}(E^i, i)$ differently from $V(E^i,i)$ to emphasize that it is a linear function, despite the fact that both take the same value for all types $i$.

Compiling the incentive compatibility, individual rationality, and obedience constraints, we can now formulate the design of the revenue-maximizing menu as the following optimization problem.  
\begin{equation}
\label{eq:GeneralOpt}
\begin{alignedat}{2}
\max_{\{\pi^i, t^i\}} \quad
& \sum_{i\in [n]} f_i t^i \\[0.5ex]
\text{s.t.}\quad
& \overline{V}(E^i, i) - t^i
    &&\geq V(E^j, i) - t^j,
    \quad \forall i,j \in [n] \\[0.5ex]
& \overline{V}(E^i, i) - t^i
    &&\geq u^i,
    \quad \forall i \in [n].
\end{alignedat}
\end{equation}
Clearly, the incentive compatibility and individual rationality constraints make an explicit appearance here. The obedience constraints are encoded implicitly in the first set of constraints whenever $i$ and $j$ coincide, because then the identical price terms on both sides cancel, and we are left with
$$ \sum_{a\in \mathcal{A}}\mathbb{E}_{\omega \sim \mu} [\pi(a \mid \omega) u^i(\omega,a)] \geq \sum_{a\in \mathcal{A}}\max_{a'\in \mathcal{A}} \mathbb{E}_{\omega \sim \mu} [\pi(a \mid \omega) u^i(\omega,a')], \qquad i\in [n],$$ which can readily be verified as being equivalent to obedience constraints, because for any function $f$, $\sum_{a\in \mathcal{A}} (f(a) - \max_{a'\in \mathcal{A}} f(a)) \geq 0$ implies that pointwise for all $a\in \mathcal{A}$ that $f(a)= \max_{a\in \mathcal{A}} f(a')$.
   Since $\overline{V}(E^i, i)$ and $V(E^j,i)$ are linear and convex, respectively, in $\pi^i$ and $\pi^j$, the formulation is a convex program. In fact, as noted in \cite{cai2021sell}, since all the convex functions that appear are piecewise linear, this is in fact a linear program.

Finally, it is important to emphasize that we have taken the buyers here to be heterogeneous in their utility functions $u^i$, but have assumed that they share a common prior distribution $\mu$ in the state. This is different from
\cite{bergemann2018design} and \cite{cai2021sell}, where the reverse is assumed: buyers have different priors, but share a utility function. (\cite{cai2021sell} also note briefly that generalizing their approach to allow buyers to vary on utility functions is straightforward.)
We conclude here by noting that, under mild conditions, our optimization problem is equivalent to that under the most general model with buyer-specific prior distributions. Specifically, if $\mu^1,\ldots,\mu^n$ are prior distributions (over $\Omega$) corresponding to each buyer type, the resulting optimization problem takes the same form of \Cref{eq:GeneralOpt}, although with expectations taken over the type-specific priors $\mu^i$ as opposed to $\mu$:
\begin{align*}
u^i
&= \max_{a \in \mathcal{A}}
    \mathbb{E}_{\omega \sim \mu^i}[u^i(\omega,a)], \\
V(E^j,i)
&= \sum_{s \in S} \max_{a\in \mathcal{A}}
    \mathbb{E}_{\omega \sim \mu^i}
    \!\bigl[\pi^j(s \mid \omega)\, u^i(\omega,a)\bigr].
\end{align*}
These expectations over $\mu^i$ can be replaced with ones over $\mu$, so long as the changes of measure $d\mu^i/d\mu$ exist (for example, if $\mu^i$ is absolutely continuous with respect to $\mu$) which can then be encoded into the utility functions. Thus, the problem with prior distributions $\mu^i$ and utility functions $u^i(\omega,a)$ is equivalent to the problem with a shared prior $\mu$ and utility functions $u^i(\omega,a)\frac{d\mu^i}{d\mu}(\omega)$.


   \vspace{1em}
\subsection{The Value of Differentiated Data Products}
Before we introduce our main results, we make an observation that underscores the benefits (and limits) of offering multiple data products on a menu. Consider three related quantities:
\begin{itemize}
    \item Let $R_\textrm{menu}$ be the revenue generated by an optimal menu.
    \item Let $R_\textrm{one}$ be the maximum revenue generated by a menu with a single experiment, in which case such an experiment can be taken to be the fully informative experiment that reveals the entire state without loss of generality. Hence, the only lever available to the seller is the price.
    \item Let $R_\textrm{full-info}$ be the revenue obtainable in an oracular setting where the buyer's type is revealed explicitly to the seller before designing the menu, and thus the seller indulges in first-degree (perfect) price discrimination. This is therefore the revenue in a setting without any seller-side information asymmetry. 
\end{itemize}
   It is easy to see that $R_\textrm{one} \leq R_\textrm{menu} \leq R_\textrm{full-info}$, and that these equalities are realized for certain settings, for example, if all types share the same utility function. We are interested in the ratios $R_\textrm{one}/R_\textrm{menu}$ that captures the benefit of offering differentiated products, and ${R_\textrm{menu}}/{R_\textrm{full-info}}$ that captures the cost of information asymmetry. To characterize these, we note the following result, where the existence clause in fact conforms to the specialized Gaussian setting we study in \Cref{sec:gaussian}. To that end, we have the following result, whose proof is in \Cref{app:value_of_diff}:

\begin{proposition} \label{prop:value_of_diff}
For any setting of type-dependent preferences, we have
$$\frac{1}{n} \le \frac{R_\mathrm{one}}{R_\mathrm{menu}} \le 1 \quad \text{and} \quad \frac{1}{n} \le \frac{R_\mathrm{menu}}{R_\mathrm{full-info}} \le 1.$$ Furthermore, for any $\varepsilon>0$, there exists a setting in which 
$$ \frac{R_\mathrm{one}}{R_\mathrm{menu}} \le \frac{1}{n} + \epsilon \quad \text{and} \quad \frac{R_\mathrm{menu}}{R_\mathrm{full-info}} = 1$$
hold simultaneously.
\end{proposition}
The key takeaway from this result is that the ability to provide multiple data products can be a substantial benefit, even over an optimally priced single product, and, in certain settings, completely overcome the information asymmetry to which the seller is subject. 

\section{Near-optimal Menus for Large State Spaces}
\label{sec:sampleLP}
We are now prepared to state our main result. Recall that $n$ and $m$ denote, respectively, the number of buyer types and actions. Our proposed algorithm is given in Algorithm~\ref{alg:algo}.

\begin{algorithm}    
\caption{Lazy Experiments via Linear Programming}
\label{alg:algo}
    \textbf{Input:} Buyer's type $i\in [n]$, the realized state $\overline{\omega}$.  \\
    \textbf{Parameters:} Sample budget $K$.
    \begin{algorithmic}[1]
    	\State Draw $K-1$ samples $\omega_{2:K} \sim \mu$, and set $\omega_1 := \overline{\omega}$.
    	\State Solve the linear program below to obtain an optimal solution $\{\pi^i, t^i\}$.
    \begin{align}
    \max &\sum_{i \in [n]} f_it^i \label{lp:cai}\\
    \text{s.t.} &\sum_{a \in \mathcal{A}} \dfrac{1}{K} \sum_{k = 1}^K \pi^i(a \mid \omega_k)u^i(a, \omega_k) - t^i \geq  \sum_{a \in \mathcal{A}} v_{a,i,j} - t^j& \  \forall i,j \in [n]\label{lp:ic} \\
    & v_{a,i,j} \geq \dfrac{1}{K} \sum_{k = 1}^K \pi^j(a \mid \omega_k)u^i(a', \omega_k) & \forall a,a' \in \mathcal{A}, \; i , j \in [n]\nonumber\\
    & \sum_{a \in \mathcal{A}} \dfrac{1}{K} \sum_{k = 1}^K \pi^i(a \mid \omega_k)u^i(a, \omega_k) - t^i \geq  \max_{a \in \mathcal{A}} \frac{1}{K}\sum_{k=1}^Ku^i(a,\omega_k), & \forall i \in [n]\label{lp:ir}\\
    &\sum_{a \in \mathcal{A}}\pi^{i}(a\mid \omega_k)= 1, &\forall i \in [n], k\in [K]\nonumber\\
    &\pi^{i}(a\mid \omega_k) \geq 0, &\forall i\in [n], k\in [K], a\in \mathcal{A}\nonumber
	\end{align}
    	\State Output a signal sampled from $\pi^i(\cdot \mid \overline{\omega})$ and the corresponding price $t^i$.
\end{algorithmic}
\end{algorithm}

\begin{remark}\label{rem:order}
The linear program (LP) in \Cref{lp:cai} is naturally symmetric with respect to its inputs $\omega_{1:K}$. We assume that the LP solver respects this symmetry, that is, if the solver is run with two different permutations of $\omega_{1:K}$, for any fixed $a\in \mathcal{A}$ and $\omega\in \omega_{1:K}$, the output $\pi^i(a \mid \omega)$, or its distribution if the outputs are random, is identical. This can always be ensured by first randomly permuting $\omega_{1:K}$ before the LP is formed.
\end{remark}

Our main guarantee for \Cref{alg:algo} is as follows.

\begin{theorem} \label{thm:main}
	Let $E^i$ be the experiment that \Cref{alg:algo} executes for input type $i$ with sample budget $K$, and let $t^i$ be the corresponding (randomized) price.
    For any \(\varepsilon,\delta\in(0,1)\), there exists K satisfying \[K = O\left(\min\left\lbrace\dfrac{n^2}{\varepsilon^2}, \frac{1}{\varepsilon^4}\right\rbrace\left(m\log \frac{mn}{\delta}\right)\right)\] such that $\{E^i, t^i\}$ is an incentive-compatible direct menu whose revenue is within $\varepsilon$ of the optimal with probability at least $1 - \delta$, that is, if $R^*$ is the revenue of an optimal menu and $\textrm{Rev}(\mathcal{A})$ is the revenue associated with \Cref{alg:algo} averaged over the buyer's type distribution, we have \[ \textrm{Rev}(\mathcal{A}) \;\ge\; \textrm{R*} - \varepsilon.\]
\end{theorem}

We also briefly remark that the $O(1/\varepsilon^2)$ scaling of the sample complexity in $\varepsilon$ is optimal for any sample-based algorithm. We give an elementary proof of this claim in \Cref{sec:lb}. 

\begin{proposition}\label{prop:lb}
Let $R^*$ be the revenue of an optimal menu. Choose any $\varepsilon \in (0,0.1)$ and $\delta\in (0,1)$. Any algorithm $\mathcal{A}$ that guarantees a revenue within $\varepsilon$ of the optimal, that is, $\textrm{Rev}(\mathcal{A})\geq R^*-\varepsilon$, with probability at least $1-\delta$, and whose sole mode of accessing $\mu$ is through a sampling oracle, must sample $\Omega(\log (1/\delta)/\varepsilon^2)$ samples.
\end{proposition}

The key difference in \Cref{alg:algo}, compared to the linear program (\Cref{eq:GeneralOpt}), is that it does not produce explicit representations of the experiments, which live in $(\Delta \mathcal{A})^\mathcal{S}$ and hence whose sizes are exponential in the number of states. Instead, Algorithm~\ref{alg:algo} encodes an implicit representation of the experiment, along with a stochastic price, that a buyer would receive if they declare their type as being $i$. Such an implicit representation is sufficient insofar as it is able to sample signals for  just the realized state. In computational terms, the algorithm simply solves the linear program from \cite{cai2021sell} internally while restricting the state space to a small randomly chosen subset $S \subseteq \Omega$ pretending that the observed empirical distribution is the common prior. This computation can be performed in time polynomial in $m,n$ and $1/\varepsilon$, which crucially does not scale with the size of the state space. (Technically, {\em naming} or addressing any object taking $N$ possible values requires $\log N$ space and time in the worst case, but we treat this as a unit cost operation.) Since the realized state is part of $S$ by design, this allows us to generate the signal for it. We show that  for such an implicit menu, it is in the interest of the buyer to truthfully report their type. Notice that the price extracted depends on the state, but the way this is generated is agreed to when the buyer chooses an experiment-price pair. 

Although random sampling here is reminiscent of sample average approximation methods~\citep{saa-paper}, one cannot hope that the resolution of this sampling is such that the entire menu can be reconstructed, instead we show that a few samples are enough to {\em locally decode} the linear program and sample the signals for the realized state. This approach is at the center of many results in information and learning theory, for example, locally decodable error-correcting codes \citep{reed1953class}, transductive learning \citep{haussler1994predicting,kakade2005batch} and multiple-kernel learning \citep{lanckriet2004learning}.

More sophisticated variants of this argument occur in the literature on mechanism design \citep{hartline2011bayesian,cai2013understanding,weinberg2014algorithms} when dealing with rich-type spaces and mechanism-to-algorithm reductions. More recently, such techniques were applied to the Bayesian persuasion setting in \cite{dughmi2016algorithmic}, where, unlike in our setting, the receiver, or in our language, the buyer, does not have private preferences. Beyond differences in the problem setting, a key qualitative distinction of our work is that the proposed mechanism is exactly incentive compatible and exactly responsive. Previous applications of similar arguments resulted in approximate incentive compatibility (for example, in \cite{hartline2011bayesian,dughmi2016algorithmic}), which means that the resultant implementations require further assumptions on the agent's (or the buyer's) behavior, specifically that it must be willing to follow suboptimal non-utility-maximizing recommendations. In contrast, in our setup, buyers are exact utility maximizers; for this, our proof of \Cref{lem:revenueloss} relies on adjusting prices to ensure exact incentive compatibility and {\em coarsening} experiments correctly in the Blackwell sense \citep{blackwell1951comparison} to ensure exact obedience.

Finally, we remark that, as presented, the randomized price extracted by the mechanism depends on the realized state, which might not be desirable (see the discussion in \cite{babaioff2012optimal}). However, this is easy to fix. Consider a variant of the mechanism in which the buyer agrees to pay a (randomized) price determined by running \Cref{alg:algo} with the input state independently sampled from $\mu$. The signal is released by running \Cref{alg:algo} separately on the realized state as usual. By linearity of expectation and the affine dependence of the constraints and the objective in the prices, the marginal distribution of the revenue of this mechanism is the same as that of the original and incentive compatibility (\Cref{lem:ICIR}) continues to hold. 

\begin{proof}{Proof of Theorem \ref{thm:main}:}
We begin by showing that the procedure constructs an incentive-compatible menu, where each type $i$ is incentivized to report its type truthfully. This hinges on the fact that the random variables sampled algorithmically, namely $\omega_{2:K}$ and the realized state, $\omega$, are identical in distribution and exchangeable a priori. As a consequence of this result, we know that the type $i$ agrees to pay the (randomized) price $t^i$.
\begin{lemma}\label{lem:ICIR}
Let $(E^i,t^i)$ be the experiment-price pair obtained by executing \Cref{alg:algo} with input type $i$. Then the menu $\{E^i, t^i\}$ implements an incentive-compatible direct mechanism.
\end{lemma}

By elementary concentration inequalities and then fixing constraints violations, we show that the LP built on sampled states has a feasible solution with an optimal value close to $R^*$. Since the prices set by the algorithm are selected to maximize the objective value, $\sum_{i=1}^{n} f_i t^i$ can only be greater.  

\begin{lemma}\label{lem:revenueloss}
    Let $R^*$ be the revenue of an optimal menu. Then there exists $K$ satisfying $K=O(\min\{{n^2}/{\varepsilon^2}, {1}/{\varepsilon^4}\}\left(m\log m{n}/{\delta}\right))$ such that, with probability at least $1-\delta$ over the realizations of $\omega_{2:K}$ and $\omega$, the LP formulated in \Cref{lp:cai} has a feasible solution $\{{\pi}^i, {t}^i\}$ with value $\sum_i f_i{t}^i \geq R^* - \epsilon$.
\end{lemma}

It remains to relate the objective value of the LP to the revenue of the algorithm. To this end, we note that although the prices $\{t^j\}_{j=1}^n$ generated via the LP depend on the realized state, they are realized independently of the buyer's type. The buyer's type only determines which of these is charged to the buyer. Thus, for any stochastic run, the algorithm earns $\sum_{j=1}^n t^j \mathbf{1}_{j=i}$, where $i$ is the buyer's type that the algorithm takes as input. Hence, the algorithm's revenue, when averaged over the type realizations, but conditioned on the internal randomness of the algorithm, that is, fixing $\omega, \omega_{2:K}$, and hence $\{t^j\}_{j=1}^n$, is
$$ \textrm{Rev}(\mathcal{A}) = \mathbb{E}_{i\sim f} \left[\sum_{j=1}^n t_j \mathbf{1}_{j=i}\omega\right] = \sum_{j=1}^n t_j \mathbb{E}_{i\sim f}[\mathbf{1}_{j=i}] = \sum_{j=1}^n f_jt^j.$$
As argued above, the latter quantity exceeds $R^*-\varepsilon$ with probability $1-\delta$.
\end{proof}

\vspace{1em}
\subsection[Proof of Lemma~\ref{lem:ICIR}]{Proof of \Cref{lem:ICIR}: Incentive Compatibility}
We wish to establish that incentive compatibility and individual rationality properties hold for this menu, specifically that for any buyer with type $i$, we want that
    \begin{align*}
    	\mathop{\mathbb{E}}_{\substack{\omega_{2:K}\sim \mu \\ \overline{\omega} \sim \mu }}\left[ \sum_{a \in \mathcal{A}}\pi^i (a \mid \overline{\omega}) u^i(a, \overline{\omega})- t^i\right] &\geq \max_j\left\lbrace \sum_{a \in \mathcal{A}}\max_{a'\in \mathcal{A}} \mathop{\mathbb{E}}_{\substack{\omega_{2:K}\sim \mu \\ \overline{\omega} \sim \mu }} [\pi^j(a \mid \overline{\omega})u^i(a', \overline{\omega}) - t^j]\right\rbrace, \\
        \mathop{\mathbb{E}}_{\substack{\omega_{2:K}\sim \mu \\ \overline{\omega} \sim \mu }}\left[ \sum_{a \in \mathcal{A}}\pi^i (a \mid \overline{\omega}) u^i(a, \overline{\omega})- t^i\right] &\geq u^i.
    \end{align*}
	for all $i$ in $[n]$. From the definition of the LP in \Cref{lp:cai}, we have for all $i$ in $[n]$ that
   \begin{align*}
   	  \sum_{a \in \mathcal{A}} \dfrac{1}{K} \sum_{k = 1}^K \pi^i(a \mid \omega_k)u^i(a, \omega_k) - t^i &\geq \max_{j}\left\lbrace\sum_{a\in \mathcal{A}} v_{a,i,j} - t^j\right\rbrace\\
   	  &\geq  \max_j\left\lbrace \sum_{a \in \mathcal{A}} \max_{a'\in \mathcal{A}}\dfrac{1}{K} \sum_{k = 1}^K \pi^j(a \mid \omega_k)u^i(a', \omega_k) - t^j\right\rbrace,\\
   	  \sum_{a \in \mathcal{A}} \dfrac{1}{K} \sum_{k = 1}^K \pi^i(a \mid \omega_k)u^i(a, \omega_k) - t^i &\geq \max_{a\in \mathcal{A}}\frac{1}{K}\sum_{k=1}^K u^i(a,\omega_k).
   \end{align*}
       
	Thus, a natural strategy is to take the expectations of these inequalities with respect to the randomness in $\omega_{1:K}$, while utilizing the fact that $\omega_{1:K}$ are identically and independently sample random variables, and hence are exchangeable. Indeed, this approach succeeds, but one has to take care that, being the result of an optimization process, $\pi^i$ itself is not fixed, but a function of $\omega_{1:K}$. Fix any $i$ and $j$ in $[n]$ and $a$ and $a'$ in $\mathcal{A}$. To deal with this, observe that
	 \begin{align*}
     \mathbb{E}_{\omega_{1:K} \sim \mu}\left[ \dfrac{1}{K} \sum_{k = 1}^K \pi^j(a \mid \omega_k)u^i(a', \omega_k)\right] &= \dfrac{1}{K} \sum_{k = 1}^K \mathbb{E}_{\omega_{1:K} \sim \mu}\left[\pi^j(a \mid \omega_k)u^i(a', \omega_k)\right] \\
        &=_{(i)} \dfrac{1}{K} \sum_{k = 1}^K \mathbb{E}_{\omega_{1:k-1}, \overline{\omega}, \omega_{k+1:K} \sim \mu}\left[\pi^j(a \mid \overline{\omega})u^i(a', \overline{\omega})\right] \\
        &=_{(ii)} \dfrac{1}{K} \sum_{k = 1}^K \mathbb{E}_{\overline{\omega}, \omega_{2:K} \sim \mu}\left[\pi^j(a \mid \overline{\omega})u^i(a', \overline{\omega})\right] \\
        &=\mathbb{E}_{\overline{\omega}\sim \mu}\mathbb{E}_{\omega_{2:K} \sim \mu}\left[\pi^j(a \mid \overline{\omega})u^i(a', \overline{\omega})\right],
    \end{align*}
    where (i) follows by renaming $\omega_k$ to $\overline{\omega}$ for each term within the sum, and (ii) utilizes the fact that all $\omega$'s are identically distributed and, hence, can be permuted. Note that in step (ii) this reordering of $\omega$'s does not alter $\pi^j(a \mid \overline{\omega})$, since although $\pi_j$ depends on the set of $\omega$'s, it is not affected by their ordering, as we have previously noted in \Cref{rem:order}. 
    
    Using this observation, we can prove the incentive compatibility property as
    \begin{align*}
    	\mathop{\mathbb{E}}_{\substack{\omega_{2:K}\sim \mu \\ \overline{\omega} \sim \mu }}\left[ \sum_{a \in \mathcal{A}}\pi^i (a \mid \overline{\omega}) u^i(a, \overline{\omega})- t^i\right] &= \sum_{a\in \mathcal{A}} \mathop{\mathbb{E}}_{\substack{\omega_{2:K}\sim \mu \\ \overline{\omega} \sim \mu }}\left[ \pi^i (a \mid \overline{\omega}) u^i(a, \overline{\omega})- t^i\right]\\
    	&=_{(i)} \sum_{a\in \mathcal{A}}\mathbb{E}_{\omega_{1:K} \sim \mu}\left[ \dfrac{1}{K} \sum_{k = 1}^K \pi^i(a \mid \omega_k)u^i(a, \omega_k) - t^i\right]\\
    	&\geq_{(ii)} \mathbb{E}_{\omega_{1:K} \sim \mu}\left[\max_j\left\lbrace \sum_{a \in \mathcal{A}} \max_{a'\in \mathcal{A}}\dfrac{1}{K} \sum_{k = 1}^K \pi^j(a \mid \omega_k)u^i(a', \omega_k) - t^j\right\rbrace\right]\\
    	&\geq_{(iii)} \max_j\left\lbrace \sum_{a \in \mathcal{A}} \max_{a'\in \mathcal{A}}\mathbb{E}_{\omega_{1:K} \sim \mu}\left[\dfrac{1}{K} \sum_{k = 1}^K \pi^j(a \mid \omega_k)u^i(a', \omega_k) - t^j\right]\right\rbrace\\
    	&=_{(i)} \max_j\left\lbrace \sum_{a \in \mathcal{A}} \max_{a'\in \mathcal{A}}\mathbb{E}_{\overline{\omega}\sim \mu}\mathbb{E}_{\omega_{2:K} \sim \mu}\left[\pi^j(a \mid \overline{\omega})u^i(a', \overline{\omega})-t^j\right]\right\rbrace,
    \end{align*}
    where all derivations labeled by (i) follow from the last observation, (ii) follows from the definition of the LP, and (iii) use the fact that for any random variable $X$, index set $\mathcal{I}$ and functions $\{f_i:i\in \mathcal{I}\}$, we have $\mathbb{E}_X \max_{i\in \mathcal{I}} f_i(X) \geq \max_{i\in \mathcal{I}}\mathbb{E}_X  f_i(X)$. The individual rationality property can be derived similarly as
        \begin{align*}
    	\mathop{\mathbb{E}}_{\substack{\omega_{2:K}\sim \mu \\ \overline{\omega} \sim \mu }}\left[ \sum_{a \in \mathcal{A}}\pi^i (a \mid \overline{\omega}) u^i(a, \overline{\omega})- t^i\right] &=_{(i)} \sum_{a\in \mathcal{A}}\mathbb{E}_{\omega_{1:K} \sim \mu}\left[ \dfrac{1}{K} \sum_{k = 1}^K \pi^i(a \mid \omega_k)u^i(a, \omega_k) - t^i\right]\\
    	&\geq_{(ii)} \mathbb{E}_{\omega_{1:K} \sim \mu}\left[\max_{a\in \mathcal{A}}\frac{1}{K}\sum_{k=1}^K u^i(a,\omega_k)\right]\\
    	&\geq_{(iii)} \max_{a\in \mathcal{A}}\mathbb{E}_{\omega_{1:K} \sim \mu}\left[\frac{1}{K}\sum_{k=1}^K u^i(a,\omega_k)\right]\\
    	&= \max_{a\in \mathcal{A}} \mathbb{E}_{\omega\sim \mu} u^i(a,\omega) = u^i,
    \end{align*}
    where the justifications for (i), (ii) and (iii) remain the same as above.
\vspace{1em}
\subsection[Proof of Lemma~\ref{lem:revenueloss}]{Proof of \Cref{lem:revenueloss}: Revenue Optimality}
We show that there exists a candidate solution to the LP described in \Cref{lp:cai} that achieves the optimal revenue, which, however, is only approximately feasible. Next, we repair the constraint violations to construct a feasible solution which achieves an ever so slightly smaller net revenue. Let us begin with the first step.

\begin{lemma}\label{lem:concentration}
Let $R^*$ be the revenue of an optimal menu. Fix any $\varepsilon,\delta>0$. For large enough $K$ satisfying $K=O((m\log {mn}{\delta}^{-1})/{\varepsilon^2})$, with probability at least $1-\delta$, there exist $\{\widetilde{\pi}^i, \widetilde{t}^i\}$ such that $\sum_i f_i\widetilde{t}^i = R^*$ and $\{\widetilde{\pi}^i, \widetilde{t}^i\}$ violates each IC (\Cref{lp:ic}) and IR (\Cref{lp:ir}) constraints by at most $\varepsilon$, while preserving other constraints in the LP (\Cref{lp:cai}).
\end{lemma}

\begin{proof}{Proof of \Cref{lem:concentration}:}
Consider an (explicit) responsive menu $\{\pi^{*i}, t^{*i}\}$ that maximizes revenue. Set $\widetilde{t}^i=t^{*i}$ and $\widetilde{\pi}^i(\cdot \mid \omega) = {\pi}^{*i}(\cdot \mid \omega)$ for all $i$ in $[n]$ and $\omega \in \omega_{1:K}$, where $\omega_{1:K}$ define the LP. In words, $\widetilde{\pi}^i$ copies the signaling scheme implied by $\pi^{*i}$, but only on the states defined in the LP. Fix any $i$ and $j$ in $[n]$ and a map $\sigma:\mathcal{A}\to\mathcal{A}$. Then, by the linearity of expectation, we have
$$\mathbb{E}_{\omega_{1:K}\sim \mu} \left[\frac{1}{K}\sum_{k=1}^K\sum_{a\in \mathcal{A}} u^i(\omega_k,\sigma(a))\widetilde{\pi}^j(a\mid \omega_k)\right] = \mathbb{E}_{\omega \sim \mu}\left[\sum_{a\in \mathcal{A}}u^i(\omega,\sigma(a)) {\pi^{*j}}(a \mid \omega)\right].$$ 
Since $\omega_{1:K}$ are independently sampled, and utilities always lie within $[0,1]$, by Hoeffding's inequality, with probability $1-\delta$, we have 
$$ \left\lvert\frac{1}{K}\sum_{k=1}^K \sum_{a\in \mathcal{A}} u^i(\omega_k,\sigma(a))\widetilde{\pi}^j(a\mid \omega_k) -  \mathbb{E}_{\omega \sim \mu}\left[\sum_{a\in \mathcal{A}}u^i(\omega,\sigma(a)) {\pi^{*j}}(a \mid \omega)\right] \right\rvert \leq \sqrt{\frac{\log 2\delta^{-1}}{2K}}.$$
Taking a union, with probability $1-\delta$, we can guarantee that the following inequality holds for all $i$ and $j$ in $[n]$ and all $m^m$ mappings from $\mathcal{A}$ to $\mathcal{A}$ simultaneously.
\begin{align}\label{eq:varep}
	\left\lvert\frac{1}{K}\sum_{k=1}^K \sum_{a\in \mathcal{A}} u^i(\omega_k,\sigma(a))\widetilde{\pi}^j(a\mid \omega_k) -  \mathbb{E}_{\omega \sim \mu}\left[\sum_{a\in \mathcal{A}}u^i(\omega,\sigma(a)) {\pi^{*j}}(a \mid \omega)\right] \right\rvert \leq \underbrace{\sqrt{\frac{m\log 2mn\delta^{-1}}{K}}}_{= \varepsilon/2}
\end{align}
The value of $K$ is chosen so that this upper bound is at most $\varepsilon/2$.

Since $\{\pi^*, t^*\}$ is a responsive menu, we also know for all $i$ in $n$ that 
\begin{align}
\mathbb{E}_{\omega \sim \mu}\left[\sum_{a\in \mathcal{A}}u^i(\omega,a) {\pi^{*i}}(a \mid \omega)\right] - t^{*i} &\geq \max_j\left\lbrace \sum_{a\in \mathcal{A}} \max_{a'\in \mathcal{A}}\mathbb{E}_{\omega \sim \mu}\left[u^i(\omega,a) {\pi^{*j}}(a' \mid \omega)\right] - t^{*j} \right\rbrace \nonumber\\
&= \max_j\left\lbrace \max_{\sigma: \mathcal{A} \to\mathcal{A}} \mathbb{E}_{\omega \sim \mu}\left[\sum_{a\in \mathcal{A}}u^i(\omega,\sigma(a)) {\pi^{*j}}(a \mid \omega)\right] - t^{*j} \right\rbrace,\label{eq:icforstar}\\
\mathbb{E}_{\omega \sim \mu}\left[\sum_{a\in \mathcal{A}}u^i(\omega,a) {\pi^{*i}}(a \mid \omega)\right] - t^{*i} &\geq u^i,\nonumber
\end{align}
where the last equality uses the fact that for any bivariate function $f:\mathcal{A}\times\mathcal{A}\to\mathbb{R}$, we have $\sum_{a\in A}\max_{a'\in \mathcal{A}} f(a, a') = \max_{\sigma:\mathcal{A}\to\mathcal{A}} \sum_{a\in A}f(a,\sigma(a))$.

Using this observation, we can guarantee that the IC constraint is violated by at most $\varepsilon$. To see this, we first define
$$ v_{a,i,j}  =  \max_{a'\in \mathcal{A}} \frac{1}{K}\sum_{k=1}^K \widetilde{\pi}^j(a|\omega_k)u^i(a',\omega_k),$$
and then by conditioning on the event that \Cref{eq:varep} holds, observe that
\begin{align*}
	&\sum_{a \in \mathcal{A}} \dfrac{1}{K} \sum_{k = 1}^K \pi^i(a \mid \omega_k)u^i(a, \omega_k) - t^i \\
	&\geq_{(i)} \sum_{a\in \mathcal{A}}\mathbb{E}_{\omega\sim \mu}\left[u^i(\omega, a)\pi^{*i}(a\mid \omega)\right] - t^i - \frac{\varepsilon}{2}\\
	&\geq_{(ii)} \max_j\left\lbrace \max_{\sigma: \mathcal{A} \to\mathcal{A}} \mathbb{E}_{\omega \sim \mu}\left[\sum_{a\in \mathcal{A}}u^i(\omega,\sigma(a)) {\pi^{*j}}(a \mid \omega)\right] - t^{*j} \right\rbrace-\frac{\varepsilon}{2}\\
	&\geq_{(i)} \max_j\left\lbrace \max_{\sigma: \mathcal{A} \to\mathcal{A}} \sum_{a\in \mathcal{A}}\frac{1}{K}\sum_{k=1}^Ku^i(\omega_k,\sigma(a)) {\pi^{*j}}(a \mid \omega_k) - t^{*j} \right\rbrace-{\varepsilon}\\
	&= \max_j\left\lbrace \sum_{a\in \mathcal{A}}\max_{a'\in \mathcal{A}}\frac{1}{K}\sum_{k=1}^Ku^i(\omega_k,a') {\pi^{*j}}(a \mid \omega_k) - t^{*j} \right\rbrace-{\varepsilon}\\
	&=  \max_{j}\left\lbrace\sum_{a \in \mathcal{A}} v_{a,i,j} - t^j\right\rbrace -\varepsilon,
\end{align*}
where (i) and (ii) follow from \Cref{eq:varep} and \Cref{eq:icforstar}, respectively.

The proof of approximate satisfaction of the IR constraints follows similarly, except that we additionally condition on the inequality below: for all $i$ in $[n]$, we have
\begin{align*}
	\left\lvert \max_{a\in \mathcal{A}}\frac{1}{K}\sum_{k=1}^K u^i(\omega_k, a) -  u^i \right\rvert \leq \sqrt{\frac{\log 2n\delta^{-1}}{2K}},
\end{align*}
which by an application of Hoeffding's inequality and a union bound on $i$ in $[n]$ holds with probability $1-\delta$ or more. This completes the proof of \Cref{lem:concentration}.
\end{proof}

Continuing with the proof of \Cref{lem:revenueloss}, from \Cref{lem:concentration}, we know there exists a menu $\{\widetilde{\pi}^i, \widetilde{t}^i\}$ such that $\sum_{i=1}^n f_i \widetilde{t}^i = R^*$ and     
$$  \widehat{U}(\widetilde{\pi}^i, i) - \widetilde{t}^i \geq \widehat{V}(\widetilde{\pi}^j, i) - t^j - \varepsilon, \quad \text{and} \quad \widehat{U}(\widetilde{\pi}^i, i)- \widetilde{t}^i\ge u^i-\varepsilon, \quad \forall i,j \in [n], $$
where $\widehat{V}(\widetilde{\pi}^j, i) = \sum_{a\in \mathcal{A}}\max_{a'\in \mathcal{A}}\mathbb{E}_{\omega\sim \widehat{\mu}}\left[ \widetilde{\pi}^i(a \mid \omega) u^i(a',\omega)\right]$ captures the value of the experiment $\widetilde{\pi}^j$ to a buyer of type $i$ in a world where the state is drawn from the measure $\widehat{\mu}(\omega) = \frac{1}{K}\sum_{k=1}^K \mathbf{1}_{\omega=\omega_k}$, and $\widehat{U}(\widetilde{\pi}^i,i) = \mathbb{E}_{\omega\sim \widehat{\mu}}\left[\sum_{a\in \mathcal{A}} \widetilde{\pi}^i(a \mid \omega) u^i(a,\omega)\right]$ captures the utility a buyer of type $i$ gains from blindly following the recommendations of $\widetilde{\pi}^i$. Note that $\widehat{V}(\widetilde{\pi}^i,i) \geq \widehat{U}(\widetilde{\pi}^i, i)$ by definition.

Thus, we can divide the constraint violations associated with $\{\widetilde{\pi}^i, \widetilde{t}^i\}$ into two buckets: violations of the IC/IR constraints, and those of the obedience constraints. We begin by repairing the former. We offer two recipes, both giving incomparable results.

Without loss of generality, we may assume that $(\widetilde{t}^i)_{i=1}^n$ forms a nondecreasing sequence. To construct a truly feasible solution, we begin by modifying the prices to $t^i=\widetilde{t}^i-i\varepsilon$ for all $i$, thus higher indices with larger $\widetilde{t}^i$'s receive a higher discount. For ensuring participation, observe
\begin{align*}
    \widehat{V}(\widetilde{\pi}^i,i) - t^i &\geq \widehat{U}(\widetilde{\pi}^i,i) - t^i\\
    &= \widehat{U}(\widetilde{\pi}^i,i) - \widetilde{t}^i + i\varepsilon\\
    &\geq u^i - \varepsilon +i\varepsilon\\
    &\geq u^i.
\end{align*}
Although the IC constraints are not yet satisfied, we claim that with this menu, each type $i$ cannot prefer an experiment with an index lower than its own, since for all $j<i$, we have
\begin{align*}
    \widehat{V}(\widetilde{\pi}^i,i) - \widehat{V}(\widetilde{\pi}^j,i) &\geq \widehat{U}(\widetilde{\pi}^i,i) - \widehat{V}(\widetilde{\pi}^j,i)\\
    &\geq \widetilde{t}^i - \widetilde{t}^j - \varepsilon \\
    &= t^i - t^j + \varepsilon(i-j-1)\\
    &\geq t^i-t^j.
\end{align*}
Therefore, with this menu, which is possibly not incentive compatible, the seller extracts a price of $\widetilde{t}^j-j\varepsilon$, from the type $i$, for some $i\leq j\leq n$, and this in turn is at least  $\widetilde{t}^i - n\varepsilon$, since $\widetilde{t}^j$ is nondecreasing. Thus, the seller earns a net revenue that is at least $\sum_{i=1}^n f_i\widetilde{t}^i - n\varepsilon = R^*-n\varepsilon$. Finally, by relabeling the experiment that the type $i$ actually chooses in this world, where $\widehat{\mu}$ is the common prior, as $\pi^i$ and by labeling the associated price as $t^i$ we can transform this into an incentive-compatible menu with the same revenue.

A similar discounting scheme uses multiplicative discounts and guarantees a net revenue of $R^*-4\sqrt{\varepsilon}$. This uses Lemma 9 in \cite{cai2021sell}, but also see \cite{cai2021efficient,daskalakis2012symmetries} which sets $t^i = (1-\sqrt{\varepsilon})\widetilde{t}^i - \varepsilon$. For completeness, we next produce a slightly modified approach to the same effect. Recall that all prices are between $0$ and $1$. For any $\eta>0$, the interval $[0,1]$ can be partitioned into successive contiguous subintervals of length $\eta$, the first subinterval being $[0,\eta]$, the second being $[\eta, 2\eta]$, and so on, resulting in $\lceil 1/\eta\rceil$ many subintervals. Consider bucketing all $\widetilde{t}^i$'s in these subintervals based on their values. If $\widetilde{t}^i$ falls into the bucket $k$, we set $t^i=\widetilde{t}^{i}-k\varepsilon$. Now, once more, we claim that with this menu, each type $i$ cannot prefer an experiment with an original price $\widetilde{t}^{j}$ if $\widetilde{t}^{j}$'s bucket has a lower index than that assigned to $\widetilde{t}^i$. This is once again because any price in a lower bucket receives $\varepsilon$ less discount, and therefore is forbidden due to a slack of at most $\varepsilon$ in the original IC constraints. Once again, we can guarantee that the type $i$ pays at least $\widetilde{t}^i - \lceil 1/\eta\rceil \varepsilon - \eta$ price. Setting $\eta=\sqrt{\varepsilon}$ yields the claim.

 To compile the final lemma, we use the better of the two.
 
 Finally, we are ready to deal with the last obstacle: obedience constraints. We will use the forward implication of the following classical result due to \cite{blackwell1951comparison,blackwell1953equivalent}.
 
 \begin{theorem}[\cite{blackwell1951comparison}]\label{thm:black} Let $E_1=(S,\pi_1)$ and $E_2=(S,\pi_2)$ be two experiments. Then there exists a column stochastic matrix $M$ such that $\pi_2 = M\pi_1$ if and only if, regardless of the underlying utility functions, $E_2$ has less value that $E_1$.
 \end{theorem}
 
 Having repaired IC constraints, let $\widetilde{\pi}^i,\widetilde{t}^i$ be the product-price that a utility maximizing buyer of type $i$ purchases. As of yet, $\widetilde{\pi}^i$ may violate the obedience constraints. If so, there exists a map $T^i:\mathcal{A}\to\mathcal{A}$ that captures the Bayes-optimal action of type $i$ given signal $s$ from $\widetilde{\pi}^i$. Concretely, choose an arbitrary $T^i$ that 
 $$ T^i(s) \in \argmax_{a\in \mathcal{A}} \mathbb{E}_{\omega\sim \mu} \left[ \widetilde{\pi}^i(s\mid \omega)u^i(a, \omega)\right]. $$
 Construct an experiment $\pi^i=T\circ \widetilde{\pi}^i$ that first samples using $\widetilde{\pi}^i$ and then pushes the sampled state forward through $T$. The Bayes-optimal action for the type $i$ given an output $a'$ from the signaling scheme $\pi^i$ lies in
 $$ \argmax_{a\in \mathcal{A}}\left\lbrace\mathbb{E}_{\omega\sim \mu}[ \pi^i(a'\mid \omega) u^i(a,\omega)] = \sum_{s\in \mathcal{A}: T(s)=a'} \mathbb{E}_{\omega\sim \mu}[\widetilde{\pi}^i(s\mid \omega) u^i(a,\omega)]\right\rbrace. $$
 Since, by construction, $a$ is the Bayes-optimal action individually for all signals $s$ generated from $\widetilde{\pi}^i$ that corresponded to $a$ under $T$, it remains a Bayes-optimal action for the signal $a$ under $\pi^i$ also. This follows from the elementary observation that if a point minimizes a set of functions individually, it minimizes their sum, too. Thus, $\{\pi^i, \widetilde{t}^i\}$ is a responsive menu.
 
 A final check is to make sure that this post-processing of experiments does not break already established IC constraints. To this end, note that the joint distributions on the state and the Bayes-optimal action are identical under $\widetilde{\pi}^i$ and $\pi^i$, thus both experiments are equally valuable for the type $i$. However, by \Cref{thm:black}, since $\pi^i$ is a {\em garbling} (or {\em coarsening}) of $\widetilde{\pi}^i$, for any generic type $j$, the value of $\pi^i$ is at most that of $\widetilde{\pi}^i$. Thus, the IC and IR constraints in the LP are still satisfied.

 \vspace{1em}
 \subsection{Proof of \Cref{prop:lb}: Information-theoretic Limits}\label{sec:lb}
 In fact, this lower bound holds even if there is a single buyer type ($n=1$). Fix any $\varepsilon\in (0,0.1)$ and $\delta \in (0, 1)$. Consider the following setting with binary states and binary actions, that is, $\Omega=\mathcal{A}=\{0,1\}$. Let the utility function be $u(a,\omega)=\mathbf{1}_{a=\omega}$. Consider two prior distributions: $\mu_0 = \textrm{Bernoulli}(1/2)$, and $\mu_1=\textrm{Bernoulli}(1/2+c\varepsilon)$, for some $c>1$. 
 
 Without loss of generality, a revenue maximizing menu for either setting has a single experiment that completely reveals the state. This is because switching any experiment in such a menu with a fully revealing experiment does not hinder participation, and without multiple types, incentive compatibility constraints are entirely absent. The utility of the buyer given the outcome of a fully revealing experiment in both cases is $1$, while the baseline utilities are $1/2$ for $\mu_0$ and $1/2+c\varepsilon$ for $\mu_1$. Thus, the optimal revenues are $1/2$ for $\mu_0$ and $1/2-c\varepsilon$ for $\mu_1$.

 If an algorithm $\mathcal{A}$ that can solely access the underlying prior via samples guarantees a revenue within $\varepsilon$ of the optimal with probability $1-\delta$, we claim that it can be used to distinguish between the cases where the prior is $\mu_0$ and where it is $\mu_1$ with the same probability. If the price the algorithm charges is at least $1/2-\varepsilon$, we declare that the underlying prior is $\mu_0$, otherwise we declare it to be $\mu_1$. Now, we can utilize the following classical information-theoretic fact to conclude the proof, which follows by composing Le Cam's bound and the Bretagnolle-Huber inequality.

 \begin{lemma}[See e.g. \cite{polyanskiy2025information,canonne2022short}]
The number of independent samples needed to distinguish between $\textrm{Bernoulli}(1/2)$ and $\textrm{Bernoulli}(1/2+\varepsilon)$ with probability $1-\delta$ is at least $\Omega(\log (1/\delta)/\varepsilon^2)$.
 \end{lemma}
\section{Pricing High-dimensional Gaussian Data} \label{sec:gaussian}
We now turn to a special case that models the pricing and design of data products for high-dimensional data. The restriction we impose in the section allows us to construct explicit signaling schemes and derive a number of structural results that capture the nature of data markets today. 

We consider a $d$-dimensional state space $\Omega=\mathbb{R}^d$, where the state $\omega$ is drawn from a Gaussian distribution $\mu=\mathcal{N}(0,I)$. These specific choices of the mean vector and the covariance matrix do not impede on generality, because any Gaussian distribution can be transformed into a standard one via an affine map, while precomposing the loss function with the inverse of the said map. The assumption of a Gaussian prior is a strong restriction compared to the algorithmic results in \Cref{sec:sampleLP}, although if the state $\omega$ is naturally an aggregate statistic derived from many individual data points, it may be assumed to be almost Gaussian (e.g. due to central limit theorems or the asymptotic normality of MLE). The utility of the assumption is based on the simplification of the signaling scheme and the specificity of the structural results that it allows. 

The buyer here is interested in a certain aspect of this high-dimensional state, represented by $\theta_i\in \mathbb{R}^d$ for the type $i$. This vector $\theta_i$ determines the utility function of type $i$ over a real-valued action space as $u^i(\omega, a) = -(\theta_i^\top \omega-a)^2$. For completeness, we note that the utility function here is unbounded, unlike the previous sections. Thus, in this setup, the learner is interested in estimating the projection of a high-dimensional state in a certain privately known direction. The squared loss is indeed commonly found and widely used in statistics, economics, and machine learning (see, for example,~\citep{jones2020nonrivalry, acemoglu2022too, loh2011classification}).

To concretely motivate the setting, consider the case when a data seller has a dataset comprising of a large number of observations, perhaps capturing health habits, travel logs, medical history, income history and financial dealings, about some number of individuals. It is not unreasonable to imagine that given access to such characteristics, a bank could tailor the interest rate for a loan, a tour company could offer more targeted discounts, or that an insurance company could personalize the premium for health insurance. However, it is clear that these entities care about distinct yet overlapping attributes that can best inform their decisions. The bank may rely on income history and past financial status; the tour company may base its decisions on income and the propensity for travel as expressed in travel logs; the insurance company might care about health and income related attributes. This is the diversity of interests captured by $\{\theta_i\}$. If such decisions are automated through machine learning, $\theta$'s can also be interpreted as the coefficients of the linear regression model that such entities use to make business decisions.

\vspace{1em}
\subsection{Scalar Gaussian Experiments Suffice}
First, we observe that the value of any experiment is captured by its (expected) posterior covariance. To establish this, for an arbitrary experiment $E=(S,\pi)$, let $\omega_{s,E} = \mathbb{E}_{\omega\sim \mu|(s,E)}[\omega]$ be the posterior mean and $\textrm{Cov}[\omega|s,E] = \mathbb{E}_{\omega\sim \mu|(s,E)} [(\omega - \omega_{s,E}) (\omega - \omega_{s,E}) ^\top]$ be the posterior covariance upon observing the signal $s$ from experiment $E$, and thus averaging over the signal $s$ itself, let $\textrm{Cov}[\omega|E] = \mathbb{E}_{s\sim \pi\circ \mu} [\textrm{Cov}[\omega|s, E]]$ be the expected posterior covariance associated with $E$.

\begin{proposition}\label{prop:gaussval}
For any prior $\mu$ and utility functions $u^i(\omega, a) = -(\theta_i^\top \omega-a)^2$, for any experiment $E$, we have $V(E, i) = -\theta_i^\top \textrm{Cov}[\omega|E] \theta_i$.
\end{proposition}

\begin{proof}{Proof of \Cref{prop:gaussval}:}
Since $\mathbb{E}_{\omega\sim \mu|(s,E)}[\theta_i^\top \omega]=\theta_i^\top \omega_{s,E}$, we have
\begin{align*}
	u^i(s, E) &= - \min_{a\in\mathcal{A}}\mathbb{E}_{\omega\sim \mu|(s,E)}[(\theta_i^\top \omega - a)^2] \\
	&= -\mathbb{E}_{\omega\sim \mu|(s,E)}[(\theta_i^\top (\omega-\omega_{s,E}))^2] = -\theta_i^\top \textrm{Cov}[\omega|s,E]\theta_i
\end{align*}
Taking the expectation over $s$ being sampled from $\pi\circ \mu$ concludes the proof.
\end{proof}

The above proposition implies that the expected posterior covariance is a sufficient statistic to characterize the value of any experiment. We use this fact to prove that scalar Gaussian experiments are expressive enough to maximize revenue. The revenue-optimal menus we consider in this section might not be {\em responsive}.

\begin{theorem}\label{thm:gauexp}
	For a Gaussian prior $\mu=\mathcal{N}(0,I)$ and utility functions $u^i(\omega, a) = -(\theta_i^\top \omega-a)^2$, there exists a revenue maximizing menu whose each constituent experiment $E$ generates signals as $\mathcal{N}(v_E^\top \omega, \sigma_E^2)$ for some $v_E\in \mathbb{R}^d$ and some $\sigma_E^2 \in \mathbb{R}_{\geq 0}$. In addition, the same parameter choices also satisfy $V(E,i)=\frac{(v_E^\top \theta_i)^2}{\|v_E\|^2+\sigma_E^2} - \|\theta_i\|^2$.
\end{theorem}
\begin{proof}{Proof of \Cref{thm:gauexp}:}
	Let $\{E^i, t^i\}$ be a revenue-maximizing menu, which we generically assume to implement an incentive-compatible direct mechanism. For each experiment $E^i$ we will construct $v_i$ and $\sigma_i^2$ so that the scalar Gaussian experiment -- let us call it $G^i$ -- representing $\mathcal{N}(v_i^\top \omega, \sigma_i^2)$ satisfies $V(E^i, i) = V(G^i, i)$ and $V(E^i, j) \geq V(G^i,j)$ for all $i$ and $j$ in $[n]$. In words, for each experiment $E^i$, the corresponding Gaussian experiment $G^i$ is, in general, (weakly) less valuable, except that it retains its value for the intended receipt, a buyer of type $i$. Thus, $\{G^i,t^i\}$ also satisfies IC and IR constraints and is therefore a direct menu that generates the same revenue.
	
\begin{lemma}[see e.g.~\cite{box2011bayesian}]
\label{thm:multivariate}
Consider a Gaussian random vector
$X = \begin{bmatrix}
X_1 \\
X_2 \end{bmatrix}$ with mean 
$\mu = \begin{bmatrix}
\mu_1 \\
\mu_2 \end{bmatrix}$ and covariance
$\Sigma = \begin{bmatrix}
\Sigma_{11} & \Sigma_{12} \\
\Sigma_{21} & \Sigma_{22} \end{bmatrix}$, then $$(X_1|X_2 = a) \sim \mathcal{N}(\mu_1 + \Sigma_{12}\Sigma_{22}^{-1}(a-\mu_2), \Sigma_{11} - \Sigma_{12}\Sigma_{22}^{-1}\Sigma_{21}).$$
\end{lemma}

Using this folklore result, we see that $$\textrm{Cov}[\omega|G^i] = I -  \frac{v_i v_i^\top }{{\|v_i\|^2 + \sigma_i^2}}.$$ For any vector-valued random variables $Y$ and random variable $X$, we know $$\text{Cov}(Y) = \E(\text{Cov}(Y|X)) + \text{Cov}(\E(Y|X)).$$ For notational brevity, define $M_i = I-\textrm{Cov}[\omega | E^i]$, for which we know $I\succeq M_i \succeq 0$ using the previous fact. Set $v_i =  M_i\theta_i$ and $\sigma_i^2 = \theta_i^\top (M_i - M_i^2) \theta_i$, which is valid since $$M_i = M_i^{1/2} I M_i^{1/2} \succeq M_i^{1/2} M M_i^{1/2} = M_i^2.$$
Using \Cref{prop:gaussval}: 
\begin{align*}
	V(G^i, i) &=  -\theta_i^\top \textrm{Cov}[\omega|G^i] \theta_i \\
	&= \frac{(\theta_i^\top M_i\theta_i)^2}{\theta_i^\top M_i^2\theta_i + \theta_i^\top (M_i - M_i^2) \theta_i} - \|\theta_i\|^2 \\
	&= \theta_i^\top (M_i-I)\theta_i \\
	&= -\theta_i^\top \textrm{Cov}[\omega|E^i] \theta_i \\&= V(E^i, i). 
\end{align*}
Finally, we invoke \Cref{prop:gaussval} to observe that
\begin{align*}
	V(G^i, j) &=  -\theta_j^\top \textrm{Cov}[\omega|G^i] \theta_j \\
	&= \frac{(\theta_j^\top M_i\theta_i)^2}{\theta_i^\top M_i^2\theta_i + \theta_i^\top (M_i - M_i^2) \theta_i} - \|\theta_j\|^2 \\
	& = \frac{(\theta_j^\top M_i\theta_i)^2 }{\theta_i^\top M_i\theta_i} - \|\theta_j\|^2\\
	& \leq \theta_j^\top M_i\theta_j - \|\theta_j\|^2 \\
	&= \theta_j^\top (M_i-I)\theta_j \\
	&= V(E^i, j),
\end{align*}
where the sole inequality follows from the generalized Cauchy-Schwarz inequality that for any $M\succeq 0$, it holds that $(\theta_i^\top M\theta_i)(\theta_j^\top M\theta_j) \geq (\theta_i^\top M\theta_j)^2$. The addendum can be independently verified by combining the expression for $\textrm{Cov}[w|G^i]$ and \Cref{prop:gaussval}.
\end{proof}

\vspace{1em}
\subsection{Optimal Menu Design via Semidefinite Programming}
Given that there is a revenue maximizing menu composed of scalar Gaussian experiments, as we have proved above, the optimal menu design problem naturally becomes a non-convex quadratically constrained quadratic program (NC-QCQP), with the variables parametrizing the best such experiments. 

The non-convexity here arises due to the left side of IC and IR constraints (\Cref{eq:qcqp}) that involve a lower bound of a positive quadratic term. However, we prove that this NC-QCQP has an exact SDP relaxation and thus can be solved in polynomial time. In general, understanding when an NC-QCQP admits an exact, or even an approximate, SDP relaxation is a challenging area of study in convex optimization \citep{fradkov1979thes,burer2020exact,nesterov1997quality,wang2022tightness}. Our proof, on the other hand, follows entirely via elementary inequalities. The key observation used here is that each positive quadratic term involves a disjoint set of variables, and its every appearance as an upper bound uses the same coefficients.

\begin{proposition}\label{prop:qcqp}
Let $\{v_i,t^i\}$ be an optimal solution to the following (non-convex) QCQP.
	\begin{equation}
    \label{eq:qcqp}
    \begin{array}{ll@{}ll}
    \max\limits_{\{v_i, t^i\}}   & \displaystyle \sum_{i\in [n]} f_i t^i & \\
        \text{subject to}&  (\theta_i^\top v_i)^2 - t^i \geq    (\theta_i^\top v_j)^2 - t^j, &\quad \forall i,j \in [n]\times[n]\\
    &  (\theta_i^\top v_i)^2- t^i\ge 0 , &\quad \forall i\in [n]\\
    &  \|v_i\|^2 \leq 1, &\quad \forall i\in [n]\\
    \end{array}
 \end{equation}
  Then, for the setting involving a Gaussian prior $\mu=\mathcal{N}(0,I)$ and utility functions $u^i(\omega, a) = -(\theta_i^\top \omega-a)^2$, $(E^i, t^i)$ is a revenue-maximizing menu, where $E^i$ represents the experiment $\mathcal{N}(v_i^\top \omega, 1-\|v_i\|^2)$.
\end{proposition}
\begin{proof}{Proof of \Cref{prop:qcqp}:}
First, we note that any non-zero scaling of the outcome (signal) of an experiment does not alter its value because scaling operations do not change the posterior, because scaling is an invertible, and hence, information preserving, operation. Let $\mathcal{N}(v_E^\top \omega, \sigma_E^2)$ be witnesses to the statement of \Cref{thm:gauexp}. Consider $$v'_E = \frac{v_E}{\sqrt{\|v_E\|^2 +\sigma_E^2}}, \quad {\sigma'_E}^2 = \frac{\sigma_E^2}{(\|v_E\|^2 + \sigma_E^2)}.$$
Thus, without loss of generality, we can operate with Gaussian experiments (say, E) $\mathcal{N}({v'_E}^\top \omega, {\sigma'_E}^2)$ where we are guaranteed that $\|{v'_E}\|^2 + {\sigma'_E}^2 = 1$, as long as we limit $\|v'_E\|\leq 1$. Furthermore, for such Gaussian experiments, from the latter part of \Cref{thm:gauexp}, we know $V(E, i) = ({v'_E}^\top \theta_i)^2 + u^i$. Substituting this into \Cref{eq:GeneralOpt}, we arrive at the claimed QCQP. 
\end{proof}

\begin{theorem}\label{thm:sdp}
The following SDP has an optimal rank-one solution. Furthermore, given a solution $\{V_i,t^i\}$ to the SDP, $\{v_i= \frac{V_i\theta_i}{\sqrt{\theta_i^\top V_i \theta_i}}, t^i\}$ forms a solution to the NC-QCQP in \Cref{eq:qcqp}.
	\begin{equation}
    \label{eq:sdp}
    \begin{array}{ll@{}ll}
    \max\limits_{\{V_i, t^i\}}   & \displaystyle \sum_{i\in [n]} f_i t^i & \\
        \text{subject to}&  \langle V_i, \theta_i\theta_i^\top\rangle - t^i \geq    \langle V_j, \theta_i\theta_i^\top\rangle - t^j, &\quad \forall i,j \in [n]\times[n]\\
    &  \langle V_i, \theta_i\theta_i^\top\rangle- t^i\ge 0, &\quad \forall i\in [n]\\
        &  0 \preceq V_i \preceq I,  &\quad \forall i\in [n]\\
    \end{array}
 \end{equation}
\end{theorem}
\begin{proof}{Proof of \Cref{thm:sdp}:}
Let us first establish the second part of the claim. Note that
$$ (v_i^\top \theta_i)^2 = \frac{(\theta_i^\top V_i\theta_i)^2}{\theta_i^\top V_i \theta_i} = \langle V_i, \theta_i\theta_i\rangle$$ and that $$(\theta_j^\top v_i)^2 = \frac{(\theta_j^\top V_i\theta_i)^2}{\theta_i^\top V_i \theta_i} \leq \langle V_i, \theta_j\theta_j^\top\rangle $$ using the generalized Cauchy-Schwarz inequality. Moreover $$\|v_i\|^2 = \frac{\theta_i^\top V_i^2 \theta_i}{\theta_i^\top V_i \theta_i} \leq 1$$
since $0\preceq V_i\preceq I$ implies $0\preceq V_i^2 \preceq V_i$. Thus, we have shown that all the NC-QCQP constraints are satisfied by the proposed solution.

For the first part, we note that for any $\{v_i, t^i\}$ feasible for the NC-QCQP, $\{v_iv_i^\top, t^i\}$ is feasible for the SDP. Furthermore, using the second part of the claim, we know that the optimal solution to the NC-QCQP has the same objective value as the optimal solution to the SDP. 
\end{proof}

\vspace{1em}
\subsection{A Characterization of Full Surplus Extraction}
As a consequence of having private preferences, the types in our setup can extract information rents. We give a necessary and sufficient characterization  of $\theta_i$'s under which the seller can perfectly screen the buyer and extract the entire revenue as if the buyer's type were publicly revealed. Our characterization, {\em well separatedness}, is an intuitive condition capturing that the preference vectors are sufficiently non-overlapping, which in turn avoids incentive compatibility conflicts. In the absence of this information asymmetry, the seller's revenue is $\sum_{i=1}^n f_i \|\theta_i\|^2$. Without loss of generality, we assume henceforth that $f_i > 0$ for all types $i$.

\begin{theorem}\label{thm:surplus}
Consider a setting with a Gaussian prior $\mu=\mathcal{N}(0,I)$ and utility functions $u^i(\omega, a) = -(\theta_i^\top \omega-a)^2$. The seller extracts the full surplus from the buyer, with revenue totaling $\sum_{i=1}^n f_i \|\theta_i\|^2$, if and only if $\{\theta_i\}$ is well separated, concretely that for all $i$ and $j$ in $[n]$, we have $|\theta_i^\top \theta_j| \leq \|\theta_i\|^2$.
\end{theorem}

\begin{proof}{Proof:}
Let us say that $\{\theta_i\}$ is well separated. Then, we claim $v_i=\theta_i/\|\theta_i\|$ and $t_i=\|\theta_i\|^2$ is feasible for the NC-QCQP in \Cref{eq:qcqp}. Note that $(\theta_i^\top v_i)^2-t^i=0$ for all $i$, while $$(\theta_i^\top v_j)^2-t^j = \frac{(\theta_i^\top \theta_j)^2}{\|\theta_j\|^2} - \|\theta_j\|^2 \leq 0.$$This establishes sufficiency.

Conversely, suppose that there exists a menu that extracts the full surplus; such a menu must be optimal. Given \Cref{thm:gauexp} and \Cref{prop:qcqp}, we can generically assume that the experiments in such an optimal menu are of the kind used whose values are linked in \Cref{eq:qcqp}. The IR constraints, together with $\|v_i\|^2 \leq 1$, imply that $v_i=\theta_i/\|\theta_i\|$. Now, for any $i$ and $j$, from the IC constraints, we arrive at $$0 \geq \frac{(\theta_i^\top \theta_j)^2}{\|\theta_j\|^2} - \|\theta_j\|^2,$$
as needed.
\end{proof}

\vspace{1em}
\subsection{Deterministic Experiments Suffice for High-dimensional Data}
Based on previous work on data pricing (for example, \cite{bergemann2018design,cai2021sell}, and indeed the previous section of this work, also), one may form the expectation that randomizing signals is crucial to extract optimal revenue. Although true from a theoretical viewpoint, a puzzling aspect of modern data markets is that most products being sold involve little to no explicit randomization. In our simplified setting, although this is not uniformly true in all dimensions, we prove that there always exist deterministic revenue-maximizing signaling schemes for high-dimensional state spaces. The proof is rooted in the intuition that, in sufficiently high-dimensions, the epistemic uncertainty inherent in the state can substitute for the aleatoric randomness imposed via explicit algorithmic randomization.

\begin{theorem}\label{thm:deter}
For a Gaussian prior $\mu=\mathcal{N}(0,I)$ and utility functions $u^i(\omega, a) = -(\theta_i^\top \omega-a)^2$, if the dimension of the state space is large enough, concretely, if $d \geq n$, there exists a revenue-maximizing menu composed entirely of deterministic experiments.
\end{theorem}

\begin{proof}{Proof of \Cref{thm:deter}:}
Consider a setting where $d\geq n$ and an optimal solution composed of scalar Gaussian experiments of the form $\mathcal{N}(v_i^\top \omega, 1-\|v_i\|^2)$ where $\{v_i\}$ are derived from an optimal solution to the NC-QCQP in \Cref{eq:qcqp}. 

Either $\|v_i\|=1$ for all $i$, in which case we are done, since the variance of the scalar Gaussian experiments $1-\|v_i\|^2$ is zero. If not, there exists some $i$ with $\|v_i\|<1$. Since $n-1<d$, there must exist a nonzero vector $\Delta_i\in \mathbb{R}^d$ that is perpendicular to $\{\theta_j: j\neq i\}$ and $\theta_i^\top v_i\theta_i^\top \Delta_i\geq 0$; the latter can be ensured by flipping the sign of $\Delta_i$ if necessary. 

Consider setting $v'_i = v_i + \alpha \Delta_i$ where $\alpha$ is nonnegative and set so that $\|v'_i\|=1$, while retaining all other variables unchanged in the NC-QCQP. We claim the solution remains feasible, and in doing so, we haven't altered the transfers. To see this, note that $(\theta_j^\top v'_i) = (\theta_j^\top v_i)^2$ for all $j\neq i$, and hence the other types retain their opinion of $v_i$ (or correctly, $E^i$). However, for type $i$, we get $$(\theta_i^\top v'_i)^2 = (\theta_i^\top v_i)^2 + (\theta_i^\top \Delta_i)^2 + 2 \theta_i^\top v'_i \theta_i^\top \Delta_i \geq (\theta_i^\top v_i)^2,$$
and thus the new experiment $E^i$ is at least as attractive to type $i$ as the one being replaced. This procedure can be repeated until all $\{v_i\}$ have unit norms.
\end{proof}
\section{Experimental Validation}\label{sec:experiments}
To validate the algorithmic developments in the paper, we perform numerical experiments to show that \Cref{alg:algo} achieves a near-optimal revenue given a small number of samples from the prior. 


\begin{figure}[h!]
    \centering
    \includegraphics[scale=0.5]{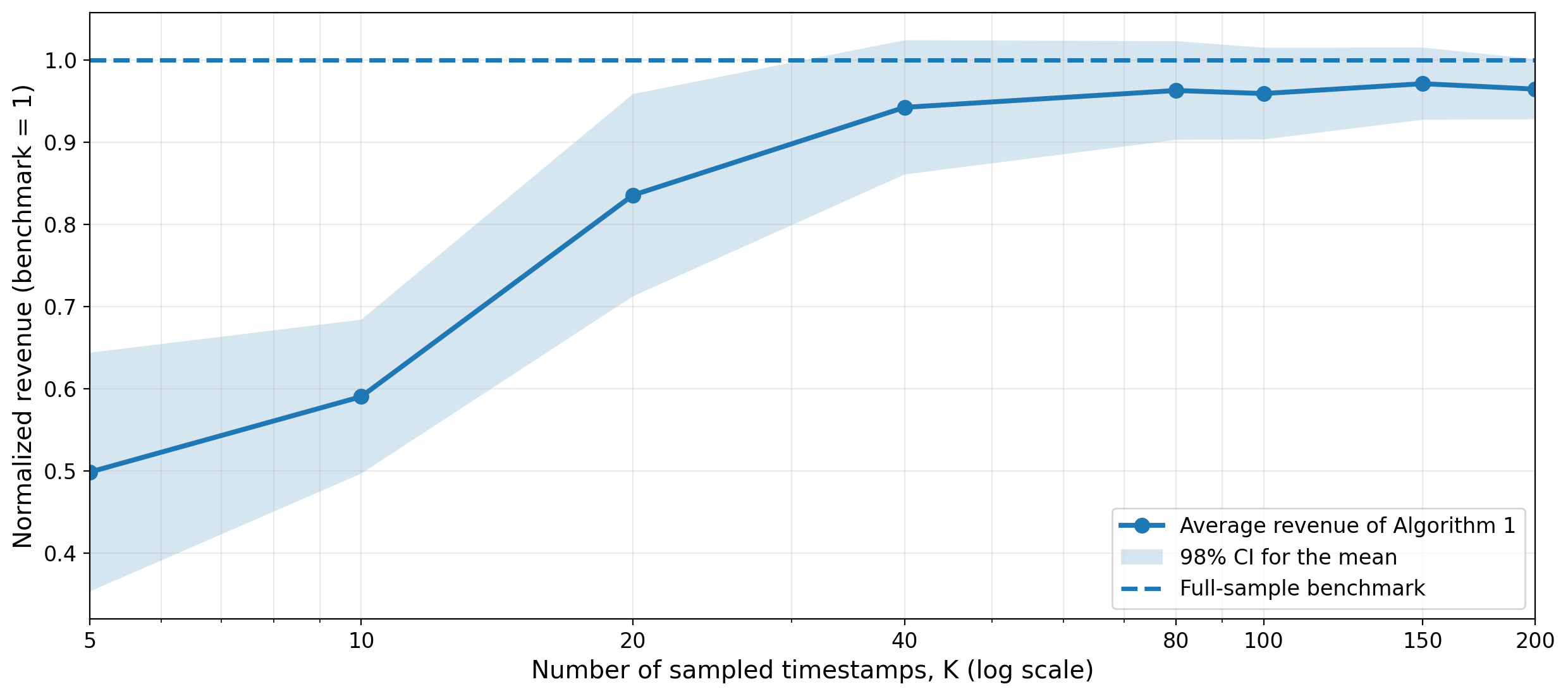}
    \caption{The ratio of revenue of \Cref{alg:algo} to the optimal revenue, plotted as a function of the number of samples made available to \Cref{alg:algo}. The solid line represents the mean, and the band captures a 98\% confidence interval. Note that the x-axis is displayed on a logarithmic scale. The performance of \Cref{alg:algo} quickly approaches the optimal revenue given a small number of samples.}
    \label{fig:1}
\end{figure}
To benchmark \Cref{alg:algo}, we consider a setting where the underlying state of the world encodes travel times between all edges in a connected undirected graph, which is unknown to data buyers. The prior on this state is public knowledge; in our experiments, it is reconstructed from real historical data, as we will discuss soon. Buyers, or more concretely, their types, encode a pair of (nonidentical) vertices, say $(o,d)$, and their action space is the set of paths in this graph that begin at $o$ and end at $d$. The loss for a buyer of type $(o,d)$ taking a path $p$ is simply the time it takes to travel from $o$ to $d$ along $p$. The seller, who has access to real-time travel times (or congestion) between all edges, can therefore selectively divulge this information to buyers via a public menu that charges different prices for different pieces of information.

{\bf Dataset and Experimental Design.} To simulate our prior distribution, we use the public METR-LA traffic dataset \citep{j49q-ch56-25}, which contains $5$-minute aggregated speed measurements from $207$ highway loop detectors in the Los Angeles County highway network. The data span approximately three months (March--June 2012) and provide a time series of speed readings at each sensor. Using distances between measurement points, this can be converted into travel times. During this process, we preprocess the speeds by imputing any remaining missing values using the median, and clip all speeds below $2$ mph to avoid numerically extreme travel times. Our prior in the above setting is a uniform distribution over all these historical measurements. For our experiments, we isolate a small connected subgraph sampled from the largest connected component and use it for all experiments. This fixed subgraph has $|V|=36$ nodes and $|E|=47$ edges. 

Each buyer type is an origin--destination (OD) pair $(o_i,d_i)$, with its action set being paths between $o_i$ and $d_i$. To create a distribution over types, we sample vertex pairs without replacement which on average have long shortest-path travel times; concretely, we sample from a mid-to-long average travel time percentile band (e.g., the 60--90th percentiles). Utilities are negative travel times, affinely rescaled to be between $[0,1]$: for type $(o,d)$, path $p$, and state $\omega$ we set $u_i(\omega,p)=1-\textrm{travel}(p,\omega)/{\tau}$, for a suitably chosen universal constant $\tau$, where $\textrm{travel}(p,\omega)$ is the aggregate travel time on path $p$ arrived at by summing the individual travel times between edges encoded in $\omega$.


{\bf Evaluation Protocol.} For each sample size $k$, we draw $k$ samples from prior, constructed described above, and use these to execute \Cref{alg:algo}. We use Gurobi~\citep{gurobi} to solve the linear program. Confidence bands and means, shown in \Cref{fig:1} are calculated by repeating this experiment $N=50$ times. We normalize the revenue obtained by \Cref{alg:algo} by the optimal revenue. In general, computing the optimal revenue exactly over high-dimensional state spaces is intractable; our state space is 47-dimensional. But, as an exception, specifically because we chose a prior that is an empirical distribution over a finite support, and not a continuous distribution, we can run the full linear program from \cite{cai2021sell} to benchmark our algorithm.

{\bf Results.} In \Cref{fig:1}, where it is worth emphasizing the horizontal axis is on a logarithmic scale, we note that the revenue obtained by \Cref{alg:algo} quickly approaches the optimal, with increasing sample budget. Notably, it attains about 95\% of the optimal revenue given just 40 samples.

\section{Conclusion}
\label{sec:conclusion}

We revisited the information-selling framework of~\cite{bergemann2018design} and its algorithmic study in~\cite{cai2021sell}, noting that the size of the state space grows exponentially when used to model the sale of high-dimensional data, and proposed an algorithm that produces a near-optimal menu whose sample complexity and runtime are independent of the size of the state space. We analyzed the special case of Gaussian data, for which we found a compact description of the space of data products, an efficient algorithm to compute an optimal menu, and an intuitive necessary and sufficient condition for full surplus extraction.

\newpage
\bibliographystyle{plainnat}
\bibliography{allcite}

@article{arora2005pricing,
  title={Pricing diagnostic information},
  author={Arora, Ashish and Fosfuri, Andrea},
  journal={Management Science},
  volume={51},
  number={7},
  pages={1092--1100},
  year={2005},
  publisher={INFORMS}
}

@article{sundararajan2004nonlinear,
  title={Nonlinear pricing of information goods},
  author={Sundararajan, Arun},
  journal={Management science},
  volume={50},
  number={12},
  pages={1660--1673},
  year={2004},
  publisher={INFORMS}
}

@article{neumann2019effective,
  title={How effective is third-party consumer profiling and audience delivery?: Evidence from field studies},
  author={Neumann, Nico and Tucker, Catherine E and Whitfield, Timothy},
  journal={Evidence from Field Studies (May 16, 2019). Forthcoming in Marketing Science-Frontiers},
  year={2019}
}

@article{guo2025selling,
  title={Selling Data to Marketers},
  author={Guo, Liang},
  journal={Management Science},
  year={2025},
  publisher={INFORMS}
}

@article{saa-paper,
author = {Kleywegt, Anton J. and Shapiro, Alexander and Homem-de-Mello, Tito},
title = {The Sample Average Approximation Method for Stochastic Discrete Optimization},
journal = {SIAM Journal on Optimization},
volume = {12},
number = {2},
pages = {479-502},
year = {2002}
}

@techreport{reed1953class,
  title={A class of multiple-error-correcting codes and the decoding scheme},
  author={Reed, Irving S},
  year={1953}
}

@article{haussler1994predicting,
  title={Predicting $\{$0, 1$\}$-functions on randomly drawn points},
  author={Haussler, David and Littlestone, Nick and Warmuth, Manfred K},
  journal={Information and Computation},
  volume={115},
  number={2},
  pages={248--292},
  year={1994},
  publisher={Elsevier}
}

@article{lanckriet2004learning,
  title={Learning the kernel matrix with semidefinite programming},
  author={Lanckriet, Gert RG and Cristianini, Nello and Bartlett, Peter and Ghaoui, Laurent El and Jordan, Michael I},
  journal={Journal of Machine learning research},
  volume={5},
  number={Jan},
  pages={27--72},
  year={2004}
}

@article{kakade2005batch,
  title={From batch to transductive online learning},
  author={Kakade, Sham and Kalai, Adam T},
  journal={Advances in Neural Information Processing Systems},
  volume={18},
  year={2005}
}

@article{bergemann2018design,
  title={The design and price of information},
  author={Bergemann, Dirk and Bonatti, Alessandro and Smolin, Alex},
  journal={American Economic Review},
  volume={108},
  number={1},
  pages={1--48},
  year={2018},
  publisher={American Economic Association 2014 Broadway, Suite 305, Nashville, TN 37203}
}

@inproceedings{cai2021sell,
  title={How to Sell Information Optimally: An Algorithmic Study},
  author={Cai, Yang and Velegkas, Grigoris},
  booktitle={Proceedings of the12th Innovations in Theoretical Computer Science Conference},
  volume={185},
  year={2021}
}

@article{admati1986monopolistic,
  title={A monopolistic market for information},
  author={Admati, Anat R and Pfleiderer, Paul},
  journal={Journal of Economic Theory},
  volume={39},
  number={2},
  pages={400--438},
  year={1986},
  publisher={Elsevier}
}

@inproceedings{babaioff2012optimal,
  title={Optimal mechanisms for selling information},
  author={Babaioff, Moshe and Kleinberg, Robert and Paes Leme, Renato},
  booktitle={Proceedings of the 13th ACM Conference on Electronic Commerce},
  pages={92--109},
  year={2012}
}

@article{myerson1982optimal,
  title={Optimal coordination mechanisms in generalized principal--agent problems},
  author={Myerson, Roger B},
  journal={Journal of mathematical economics},
  volume={10},
  number={1},
  pages={67--81},
  year={1982},
  publisher={Elsevier}
}

@book{box2011bayesian,
  title={Bayesian inference in statistical analysis},
  author={Box, George EP and Tiao, George C},
  year={2011},
  publisher={John Wiley \& Sons}
}

@article{myerson1981optimal,
  title={Optimal auction design},
  author={Myerson, Roger B},
  journal={Mathematics of Operations Research},
  volume={6},
  number={1},
  pages={58--73},
  year={1981},
  publisher={INFORMS}
}

@inproceedings{blackwell1951comparison,
  title={Comparison of experiments},
  author={Blackwell, David},
  booktitle={Proceedings of the Second Berkeley symposium on Mathematical Statistics and Probability},
  volume={2},
  pages={93--103},
  year={1951},
  organization={University of California Press}
}

@article{blackwell1953equivalent,
  title={Equivalent comparisons of experiments},
  author={Blackwell, David},
  journal={The Annals of Mathematical Statistics},
  pages={265--272},
  year={1953},
  publisher={JSTOR}
}

@article{riley1983optimal,
  title={Optimal selling strategies: When to haggle, when to hold firm},
  author={Riley, John and Zeckhauser, Richard},
  journal={The Quarterly Journal of Economics},
  volume={98},
  number={2},
  pages={267--289},
  year={1983},
  publisher={MIT Press}
}

@article{loh2011classification,
  title={Classification and regression trees},
  author={Loh, Wei-Yin},
  journal={Wiley interdisciplinary reviews: data mining and knowledge discovery},
  volume={1},
  number={1},
  pages={14--23},
  year={2011},
  publisher={Wiley Online Library}
}

@article{acemoglu2022too,
  title={Too much data: Prices and inefficiencies in data markets},
  author={Acemoglu, Daron and Makhdoumi, Ali and Malekian, Azarakhsh and Ozdaglar, Asu},
  journal={American Economic Journal: Microeconomics},
  volume={14},
  number={4},
  pages={218--256},
  year={2022},
  publisher={American Economic Association 2014 Broadway, Suite 305, Nashville, TN 37203-2425}
}

@article{jones2020nonrivalry,
  title={Nonrivalry and the Economics of Data},
  author={Jones, Charles I and Tonetti, Christopher},
  journal={American Economic Review},
  volume={110},
  number={9},
  pages={2819--2858},
  year={2020},
  publisher={American Economic Association 2014 Broadway, Suite 305, Nashville, TN 37203}
}

@article{esHo2007optimal,
  title={Optimal information disclosure in auctions and the handicap auction},
  author={Es{\H{o}}, P{\'e}ter and Szentes, Balazs},
  journal={The Review of Economic Studies},
  volume={74},
  number={3},
  pages={705--731},
  year={2007},
  publisher={Wiley-Blackwell}
}

@inproceedings{devanur2020optimal,
  title={Optimal mechanism design for single-minded agents},
  author={Devanur, Nikhil R and Goldner, Kira and Saxena, Raghuvansh R and Schvartzman, Ariel and Weinberg, S Matthew},
  booktitle={Proceedings of the 21st ACM Conference on Economics and Computation},
  pages={193--256},
  year={2020}
}

@inproceedings{giannakopoulos2014duality,
  title={Duality and optimality of auctions for uniform distributions},
  author={Giannakopoulos, Yiannis and Koutsoupias, Elias},
  booktitle={Proceedings of the fifteenth ACM conference on Economics and computation},
  pages={259--276},
  year={2014}
}

@inproceedings{haghpanah2015reverse,
  title={Reverse mechanism design},
  author={Haghpanah, Nima and Hartline, Jason},
  booktitle={Proceedings of the sixteenth ACM conference on Economics and Computation},
  pages={757--758},
  year={2015}
}

@inproceedings{fiat2016fedex,
  title={The fedex problem},
  author={Fiat, Amos and Goldner, Kira and Karlin, Anna R and Koutsoupias, Elias},
  booktitle={Proceedings of the 2016 ACM Conference on Economics and Computation},
  pages={21--22},
  year={2016}
}

@article{kamenica2011bayesian,
  title={Bayesian persuasion},
  author={Kamenica, Emir and Gentzkow, Matthew},
  journal={American Economic Review},
  volume={101},
  number={6},
  pages={2590--2615},
  year={2011},
  publisher={American Economic Association}
}

@article{kamenica2019bayesian,
  title={Bayesian persuasion and information design},
  author={Kamenica, Emir},
  journal={Annual Review of Economics},
  volume={11},
  pages={249--272},
  year={2019},
  publisher={Annual Reviews}
}

@article{alonso2016bayesian,
  title={Bayesian persuasion with heterogeneous priors},
  author={Alonso, Ricardo and C{\^a}mara, Odilon},
  journal={Journal of Economic Theory},
  volume={165},
  pages={672--706},
  year={2016},
  publisher={Elsevier}
}

@inproceedings{dughmi2016algorithmic,
  title={Algorithmic bayesian persuasion},
  author={Dughmi, Shaddin and Xu, Haifeng},
  booktitle={Proceedings of the forty-eighth annual ACM symposium on Theory of Computing},
  pages={412--425},
  year={2016}
}

@article{admati1990direct,
  title={Direct and indirect sale of information},
  author={Admati, Anat R and Pfleiderer, Paul},
  journal={Econometrica: Journal of the Econometric Society},
  pages={901--928},
  year={1990},
  publisher={JSTOR}
}

@article{liu2021optimal,
  title={Optimal pricing of information},
  author={Liu, Shuze and Shen, Weiran and Xu, Haifeng},
  journal={arXiv preprint arXiv:2102.13289},
  year={2021}
}

@inproceedings{chen2022selling,
  title={Selling data to a machine learner: Pricing via costly signaling},
  author={Chen, Junjie and Li, Minming and Xu, Haifeng},
  booktitle={International Conference on Machine Learning},
  pages={3336--3359},
  year={2022},
  organization={PMLR}
}

@inproceedings{cai2021efficient,
  title={An Efficient $\epsilon$-BIC to BIC Transformation and Its Application to Black-Box Reduction in Revenue Maximization},
  author={Cai, Yang and Oikonomou, Argyris and Velegkas, Grigoris and Zhao, Mingfei},
  booktitle={Proceedings of the 2021 acm-siam symposium on discrete algorithms (soda)},
  pages={1337--1356},
  year={2021},
  organization={SIAM}
}

@inproceedings{daskalakis2012symmetries,
  title={Symmetries and optimal multi-dimensional mechanism design},
  author={Daskalakis, Constantinos and Weinberg, Seth Matthew},
  booktitle={Proceedings of the 13th ACM conference on Electronic commerce},
  pages={370--387},
  year={2012}
}

@inproceedings{hartline2011bayesian,
  title={Bayesian incentive compatibility via matchings},
  author={Hartline, Jason D and Kleinberg, Robert and Malekian, Azarakhsh},
  booktitle={Proceedings of the twenty-second annual ACM-SIAM symposium on Discrete Algorithms},
  pages={734--747},
  year={2011},
  organization={SIAM}
}

@inproceedings{cai2013understanding,
  title={Understanding incentives: Mechanism design becomes algorithm design},
  author={Cai, Yang and Daskalakis, Constantinos and Weinberg, S Matthew},
  booktitle={2013 IEEE 54th Annual Symposium on Foundations of Computer Science},
  pages={618--627},
  year={2013},
  organization={IEEE}
}

@phdthesis{weinberg2014algorithms,
  title={Algorithms for strategic agents},
  author={Weinberg, S Matthew},
  year={2014},
  school={Massachusetts Institute of Technology}
}

@article{bergemann2019information,
  title={Information design: A unified perspective},
  author={Bergemann, Dirk and Morris, Stephen},
  journal={Journal of Economic Literature},
  volume={57},
  number={1},
  pages={44--95},
  year={2019},
  publisher={American Economic Association 2014 Broadway, Suite 305, Nashville, TN 37203-2425}
}

@article{bergemann2015selling,
  title={Selling cookies},
  author={Bergemann, Dirk and Bonatti, Alessandro},
  journal={American Economic Journal: Microeconomics},
  volume={7},
  number={3},
  pages={259--294},
  year={2015},
  publisher={American Economic Association 2014 Broadway, Suite 305, Nashville, TN 37203-2425}
}

@article{bergemann2022economics,
  title={The economics of social data},
  author={Bergemann, Dirk and Bonatti, Alessandro and Gan, Tan},
  journal={The RAND Journal of Economics},
  volume={53},
  number={2},
  pages={263--296},
  year={2022},
  publisher={Wiley Online Library}
}

@article{yang2022selling,
  title={Selling consumer data for profit: Optimal market-segmentation design and its consequences},
  author={Yang, Kai Hao},
  journal={American Economic Review},
  volume={112},
  number={4},
  pages={1364--1393},
  year={2022},
  publisher={American Economic Association 2014 Broadway, Suite 305, Nashville, TN 37203}
}

@article{armstrong1996multiproduct,
  title={Multiproduct nonlinear pricing},
  author={Armstrong, Mark},
  journal={Econometrica: Journal of the Econometric Society},
  pages={51--75},
  year={1996},
  publisher={JSTOR}
}

@article{rochet1998ironing,
  title={Ironing, sweeping, and multidimensional screening},
  author={Rochet, Jean-Charles and Chon{\'e}, Philippe},
  journal={Econometrica},
  pages={783--826},
  year={1998},
  publisher={JSTOR}
}

@article{bergemann2016information,
  title={Information design, Bayesian persuasion, and Bayes correlated equilibrium},
  author={Bergemann, Dirk and Morris, Stephen},
  journal={American Economic Review},
  volume={106},
  number={5},
  pages={586--591},
  year={2016},
  publisher={American Economic Association 2014 Broadway, Suite 305, Nashville, TN 37203}
}

@article{canonne2022short,
  title={A short note on an inequality between KL and TV},
  author={Canonne, Cl{\'e}ment L},
  journal={arXiv preprint arXiv:2202.07198},
  year={2022}
}

@book{polyanskiy2025information,
  title={Information theory: From coding to learning},
  author={Polyanskiy, Yury and Wu, Yihong},
  year={2025},
  publisher={Cambridge university press}
}

@article{wang2022tightness,
  title={On the tightness of SDP relaxations of QCQPs},
  author={Wang, Alex L and K{\i}l{\i}n{\c{c}}-Karzan, Fatma},
  journal={Mathematical Programming},
  volume={193},
  number={1},
  pages={33--73},
  year={2022},
  publisher={Springer}
}

@techreport{nesterov1997quality,
  title={Quality of semidefinite relaxation for nonconvex quadratic optimization},
  author={Nesterov, Yurii},
  year={1997},
  institution={Universit{\'e} catholique de Louvain, Center for Operations Research and~…}
}

@article{burer2020exact,
  title={Exact semidefinite formulations for a class of (random and non-random) nonconvex quadratic programs},
  author={Burer, Samuel and Ye, Yinyu},
  journal={Mathematical Programming},
  volume={181},
  number={1},
  pages={1--17},
  year={2020},
  publisher={Springer}
}

@article{fradkov1979thes,
  title={Thes-procedure and duality relations in nonconvex problems of quadratic programming},
  author={Fradkov, AL and Yakubovich, VA},
  journal={Vestn. LGU, Ser. Mat., Mekh., Astron,(1)},
  pages={101--109},
  year={1979}
}

@data{j49q-ch56-25,
doi = {10.21227/j49q-ch56},
url = {https://dx.doi.org/10.21227/j49q-ch56},
author = {Hengyuan He},
publisher = {IEEE Dataport},
title = {California Traffic Network Datasets: METR-LA, PEMS-BAY, PEMS04 and PEMS08 for Traffic Speed and Flow Analysis},
year = {2025} }

@misc{gurobi,
  author = {{Gurobi Optimization, LLC}},
  title = {{Gurobi Optimizer Reference Manual}},
  year = 2024,
  url = "https://www.gurobi.com"
}

\newpage
\appendix
\section{Proof of \Cref{prop:value_of_diff}: Value of Data Product Differentitation}\label{app:value_of_diff}
\vspace{1em}
First, we make an auxiliary claim that $R_\textrm{one}/R_\text{full-info} \geq 1/n$. Full surplus extraction earns a revenue $$R_\textrm{full-info} = \sum_{i=1}^n f_i (\max_{a\in \mathcal{A}} u^i(\omega, a) - u^i).$$ 
In the one-item menu, we propose to reveal the state (without noise) at a price of $$t^* = \max_{a\in \mathcal{A}} u^{i^*}(\omega, a) - u^{i^*}$$ 
where $i^*$ is the type for which $f_i (\max_{a\in \mathcal{A}} u^{i}(\omega, a) - u^{i})$ is maximized. This menu of composed of a single item generates a revenue $r^*$ that is at least $\max_{i\in [n]} f_i (\max_{a\in \mathcal{A}} u^{i}(\omega, a) - u^{i})$, since the type $i^*$ agrees with this transaction. Now
\begin{align*}
r^* &\geq \max_{i\in [n]} f_i (\max_{a\in \mathcal{A}} u^{i}(\omega, a) - u^{i})\\
&\geq \frac{1}{n}\sum_{i=1}^n f_i (\max_{a\in \mathcal{A}} u^{i}(\omega, a) - u^{i}) = \frac{1}{n}R_\textrm{full-info},
\end{align*}
establishing the auxiliary claim.

The first part of the proposition now follows directly from this claim as 
\begin{align*}
	\frac{R_{\text{one}}}{R_{\text{menu}}} &= \frac{R_{\text{one}}}{R_{\text{full-info}} }\frac{R_{\text{full-info}}}{R_{\text{menu}}} \geq \frac{R_{\text{one}}}{R_{\text{full-info}}} \geq \frac{1}{n},\\
	\frac{R_{\text{menu}}}{R_{\text{full-info}}} &= \frac{R_{\text{menu}}}{R_{\text{one}}} \frac{R_{\text{one}}}{R_{\text{full-info}}} \geq \frac{R_{\text{one}}}{R_{\text{full-info}}} \geq \frac{1}{n},
\end{align*}
since $R_\textrm{one}/R_\textrm{menu}$ and $R_\text{menu}/R_\text{full-info}$ are at most one.

For the existence part, we in fact provide a setting that conforms to the specialized Gaussian setting of \Cref{sec:gaussian}: The type $i$ with the preference vector $\theta_i = \alpha_i e_i$ occurs with probability $f_i$, where $e_i$'s are the canonical basis vectors. We will soon specify the values of $\alpha_i$ and $f_i$. Clearly
$$R_\textrm{full-info}=\sum_{i=1}^n f_i \alpha_i^2.$$ 
Similarly, we have that $$R_\textrm{menu}=\sum_{i=1}^n f_i \alpha_i^2,$$ 
which can be verified by plugging $v_i=e_i$ and $t_i = \alpha_i^2$ into the NC-QCQP in \Cref{eq:qcqp}. Now, consider an optimal one-item menu that results in revenue $R_\textrm{one}$. We can generically assume that the product on offer here reveals the state entirely since this choice cannot decrease any type's willingness to pay. Let $t^*$ be the price at which this product is sold. Now $$R_\textrm{one} = \sum_{i=1}^n f_i t^* \mathbf{1}_{\alpha_i^2 \geq t^*}.$$
Set $\alpha_i = \alpha^{-i/2}$ and $f_i = \alpha^i/\sum_{i=1}^n \alpha^i$, for some $\alpha\in (0,1]$. Now, $\sum_{i=1}^n f_i\alpha_i^2 = n / \sum_{i=1}^n \alpha^i$, while by enumerating all values of $t^*$ in $\{\alpha_j^2\}_{j=1}^n$ to determine the maximum, we have
\begin{align*}
	R_\textrm{one} &= \max_{j\in[n]} \sum_{i=1}^n \frac{\alpha^{i}}{\sum_{k=1}^n\alpha^k} \alpha^{-j} \mathbf{1}_{\alpha^{-i}\geq \alpha^{-j}} \\
	&= \frac{\max_{j\in [n]} \sum_{i=j}^n \alpha^{i-j}}{\sum_{i=1}^n \alpha^i} \\
	&< \frac{1}{(1-\alpha)\sum_{i=1}^n \alpha^i},
\end{align*}
where since $\alpha \in (0,1]$, $\alpha^{-i}\geq \alpha^{-j} \iff \alpha^{j} \geq \alpha^i \iff i \geq j$.

Thus $R_\textrm{one}/R_{\textrm{menu}} < 1/(n(1-\alpha))$. Taking a small enough $\alpha$, in particular $\alpha = \varepsilon n / (1+\varepsilon n)$, suffices.


\vspace{1em}
\section{Beyond a Gaussian Prior}
\vspace{1em}
The results in \Cref{sec:gaussian} relied on the assumption of a Gaussian prior, and it is unclear how these, or even qualitatively weaker variants, might be extended to the general case. We give two results that survive such a generalization. 

Below we state these results assuming a generic non-Gaussian prior with identity covariance, but we remark that this is without loss of any generality, because these conditions can be ensured by an affine Whitening transformation. Concretely, if $\mathbb{E}[\omega]=m$ and $\mathbb{E}[(\omega-m)(\omega-m)^\top]=\Sigma$, then $\widetilde{\omega}=\Sigma^{-1/2}(\omega-m)$ can be reaadily verified to be zero mean with identity covariance. The new utility function becomes $\widetilde{u}^i(a,\widetilde{\omega}) = u^i(a, \Sigma^{1/2}\widetilde{\omega}+m)$.

\begin{proposition}\label{lem:sdpup}
For any zero-mean prior $\mu$ with covariance identity and utility functions $u^i(\omega, a) = -(\theta_i^\top \omega-a)^2$, the optimal revenue is upper bounded by the optimal values to the SDP from \Cref{eq:sdp} and the QCQP from \Cref{eq:qcqp}. 
\end{proposition}
\begin{proof}{Proof of \Cref{lem:sdpup}}
In \Cref{thm:sdp}, we have already proved that the optimal values of the SDP and QCQP coincide. So, here, it is sufficient to prove that the optimal revenue is at most the optimal value of the SDP. 

To prove the latter claim, we will produce an explicit feasible solution to the SDP using the experiments constituting the optimal menu which we denote as $\{E^i, t^i\}$. By the revelation principle, without loss of generality, we assume that such a menu implements an incentive compatible direct mechanism. From \Cref{prop:gaussval}, we know that $V(E^i,j)=-\theta_j^\top \textrm{Cov}[\omega | E^i]\theta_j$. Now, we claim that $\{(I-\textrm{Cov}[\omega | E^i], t^i)\}$ is a feasible solution of the SDP. To this end, note that incentive compatibility of the menu implies 
\begin{align*} 
\langle I-\textrm{Cov}[\omega | E^i],\theta_i\theta_i^\top \rangle -t^i &= \|\theta_i\|^2 + V(E^i,i)-t^i \\
&\geq \|\theta_i\|^2 + V(E^j,i)-t^j \\
&= \langle I-\textrm{Cov}[\omega | E^j],\theta_i\theta_i^\top\rangle - t^j.
\end{align*}
Similarly, since $u^i = -\min_{a} \mathbb{E}_{\omega\sim \mu}(\theta_i^\top (\omega-a))^2 = -\|\theta_i\|^2 $, and hence using participation constraints from the optimal menu, we get
$$\langle I-\textrm{Cov}[\omega | E^i],\theta_i\theta_i^\top \rangle -t^i = -u^i+V(E^i,i)-t^i \geq 0.$$
Finally, recall the identity that 
$$ \textrm{Cov}(Y) = \mathbb{E}(\textrm{Cov}(Y|X)) + \textrm{Cov}(\mathbb{E}(Y|X)), $$
which implies that $0\preceq \textrm{Cov}[\omega | E^j] \preceq \mathbb{E}[\omega\omega^\top]=I$, as needed to complete the proof..
\end{proof}

\begin{proposition}\label{lem:onesidefullsurplus}
 For any zero-mean prior $\mu$ with covariance identity and utility functions $u^i(\omega, a) = -(\theta_i^\top \omega-a)^2$, having $\{\theta_i\}$'s which are well separated, that is, $\forall i,j\in [n], |\theta_i^\top \theta_j|\leq \|\theta_i\|^2$, is necessary for full surplus extraction.
\end{proposition}
\begin{proof}{Proof of \Cref{lem:onesidefullsurplus}}
Suppose there exists a menu that performs full surplus extraction, earning a revenue of $\sum_{i=1}^n f_i \|\theta_i\|^2$. Then, by \Cref{lem:sdpup}, there exists a feasible solution, say $\{(v_i,t^i)\}$, for the NC-QCQP from \Cref{eq:qcqp} with an objective value $\sum_{i=1}^n f_it^i$, which is at least $\sum_{i=1} f_i \|\theta_i\|^2$. 

Since $\|v_i\|^2\leq 1$, we get 
$ \|\theta_i\|^2 \geq (\theta_i^\top v_i)^2 \geq t^i$. Thus, $t^i=\|\theta_i\|^2$ holds for all $i$ in $[n]$. Coupled with $\|v_i\|^2\leq 1$, this in turn means that $(\theta_i^\top v_i)^2 -t^i \geq 0$ can only hold if $v_i = \pm \theta_i / \|\theta_i\|$. 

Generically, we pick $v_i = \theta_i/\|\theta_i\|$, because the NC-QCQP is in fact invariant to the sign of $v_i$. To conclude the claim, we notice that the NC-QCQP enforces 
$$ (\theta_i^\top v_i)^2 - t^i \geq (\theta_i^\top v_j)^2 - t^j.$$
We have already shown that $t^i=\|\theta_i\|^2$, hence the left side of the above inequality is zero. Substituting the values of $t^j$ and $v_j$ we have derived above, we get for all $i,j$ in $[n]$ that
$$ 0 \geq \frac{(\theta_i^\top \theta_j)^2}{\|\theta_j\|^2} -\|\theta_j\|^2, $$
which concludes our proof.
\end{proof}

\end{document}